\documentclass[a4paper,UKenglish,cleveref, autoref, thm-restate,authorcolumns]{lipics-v2019}
\usepackage{contents/macros}

\graphicspath{{./contents/graphics/}}%helpful if your graphic files are in another directory

\title{Scalable Termination Detection for Distributed Actor Systems} %TODO Please add

\titlerunning{Concurrent Termination Detection} %TODO optional, please use if title is longer than one line

\author{Dan Plyukhin}{University of Illinois at Urbana-Champaign, USA}{daniilp2@illinois.edu}{}{}%{(Optional) author-specific funding acknowledgements}%TODO mandatory, please use full name; only 1 author per \author macro; first two parameters are mandatory, other parameters can be empty. Please provide at least the name of the affiliation and the country. The full address is optional

\author{Gul Agha}{University of Illinois at Urbana-Champaign, USA}{agha@illinois.edu}{}{}%{(Optional) author-specific funding acknowledgements}

\authorrunning{D. Plyukhin and G. Agha} %TODO mandatory. First: Use abbreviated first/middle names. Second (only in severe cases): Use first author plus 'et al.'

\Copyright{Dan Plyukhin and Gul Agha} %TODO mandatory, please use full first names. LIPIcs license is "CC-BY";  http://creativecommons.org/licenses/by/3.0/

\ccsdesc[500]{Computing methodologies~Concurrent algorithms}
\ccsdesc[500]{Software and its engineering~Garbage collection}

\keywords{actors, concurrency, termination detection, quiescence detection, garbage collection, distributed systems} %TODO mandatory; please add comma-separated list of keywords

\category{} %optional, e.g. invited paper

\relatedversion{} %optional, e.g. full version hosted on arXiv, HAL, or other respository/website
\supplement{}%optional, e.g. related research data, source code, ... hosted on a repository like zenodo, figshare, GitHub, ...

\funding{This work was supported in part by the National Science Foundation under Grant No. SHF 1617401, and in part by the Laboratory Directed Research and Development program at Sandia National Laboratories, a multi-mission laboratory managed and operated by National Technology and Engineering Solutions of Sandia, LLC, a wholly owned subsidiary of Honeywell International, Inc., for the U.S. Department of Energy’s National Nuclear Security Administration under contract DE-NA0003525.}%optional, to capture a funding statement, which applies to all authors. Please enter author specific funding statements as fifth argument of the \author macro.

\acknowledgements{We would like to thank Dipayan Mukherjee, Atul Sandur, Charles Kuch, Jerry Wu, and the anonymous referees for providing valuable feedback in earlier versions of this work.}

\nolinenumbers %uncomment to disable line numbering

\EventEditors{Igor Konnov and Laura Kov\'{a}cs}
\EventNoEds{2}
\EventLongTitle{31st International Conference on Concurrency Theory (CONCUR 2020)}
\EventShortTitle{CONCUR 2020}
\EventAcronym{CONCUR}
\EventYear{2020}
\EventDate{September 1--4, 2020}
\EventLocation{Vienna, Austria}
\EventLogo{}
\SeriesVolume{2017}
\ArticleNo{44}
\begin{document}

\maketitle

\begin{abstract}
Automatic {\em garbage collection\/} (GC) prevents certain kinds of
bugs and reduces programming overhead.  GC techniques for sequential
programs are based on {\em reachability analysis\/}.
However, testing reachability from a root set is inadequate for determining whether an {\em actor\/} is garbage because an unreachable actor 
may send a message to a reachable actor.  
Instead, it is sufficient to check \emph{termination} (sometimes also called \emph{quiescence}): an
actor is terminated if it is not currently processing a message 
and cannot receive a message in the future.
Moreover, many actor frameworks provide all
actors with access to file I/O or external storage; without inspecting an actor's internal code, it is necessary to check that the actor has terminated to ensure that it may be garbage collected in these frameworks.
Previous algorithms to detect actor garbage require coordination
mechanisms such as causal message delivery or nonlocal monitoring of
actors for mutation.  Such coordination mechanisms adversely affect
concurrency and are therefore expensive in distributed systems.  We
present a low-overhead {\em reference listing\/} technique (called
{\em DRL\/}) for termination detection in actor systems.  DRL is based
on asynchronous local snapshots and message-passing between actors.
This enables a decentralized implementation and transient network
partition tolerance.  The paper provides a formal description of DRL,
shows that all actors identified as garbage have indeed terminated
(safety), and that all terminated actors--under certain reasonable
assumptions--will eventually be identified (liveness).

% : given an arbitrary
% temporally distributed set of actor-local snapshots, a GC process can
% either deduce that some actors have terminated or request additional
% snapshots from non-terminated actors. Because snapshots do not need to
% be synchronized and network partitions are tolerated, our algorithm
% can be implemented in existing distributed actor systems with lower
% overhead than other approaches.
\end{abstract}

\section{Introduction}

The actor model~\cite{books/daglib/0066897,journals/cacm/Agha90} is a
foundational model of concurrency that has been widely adopted for its
scalability: for example, actor languages have been used to implement services at
PayPal~\cite{PayPalBlowsBillion}, Discord~\cite{vishnevskiyHowDiscordScaled2017}, and in the United Kingdom's National Health Service
database~\cite{NHSDeployRiak2013}.  In the actor model, stateful processes known as \emph{actors}
execute concurrently and communicate by
sending asynchronous messages to other actors, provided they have a \emph{reference} (also called a \emph{mail address} or \emph{address} in the literature) to the recipient. Actors can also spawn new actors. An actor is said to be \emph{garbage} if it can be destroyed without affecting the system's observable behavior.

Although a number of algorithms for automatic actor GC have been proposed \cite{ clebschFullyConcurrentGarbage2013,
  kafuraConcurrentDistributedGarbage1995,
  vardhanUsingPassiveObject2003,
  venkatasubramanianScalableDistributedGarbage1992,
  wangConservativeSnapshotbasedActor2011,
  wangDistributedGarbageCollection2006}, actor languages and
frameworks currently popular in industry (such as Akka \cite{Akka},
Erlang \cite{armstrongConcurrentProgrammingERLANG1996}, and Orleans
\cite{bykovOrleansCloudComputing2011}) require that programmers
garbage collect actors manually.  We believe this is because the algorithms
proposed thus far are too expensive to implement in distributed
systems.  In order to find applicability in real-world actor runtimes, we argue that a GC algorithm should satisfy the following properties:

\begin{enumerate}
\item (\emph{Low latency}) GC should not restrict concurrency in the
  application.
\item (\emph{High throughput}) GC should not impose significant space
  or message overhead.
\item (\emph{Scalability}) GC should scale with the number of actors and nodes in the system.
\end{enumerate} 

To the best of our knowledge, no previous algorithm satisfies all
three constraints. The first requirement precludes any global
synchronization between actors, a ``stop-the-world'' step, or a
requirement for causal order delivery of all messages.  The second
requirement means that the number of additional ``control'' messages imposed by the
algorithm should be minimal.  The third requirement precludes
algorithms based on global snapshots, since taking a
global snapshot of a system with a large number of nodes is
infeasible.

To address these goals, we have developed a garbage collection
technique called \emph{DRL} for \emph{Deferred Reference Listing}. The primary advantage of DRL is that it is decentralized and incremental: local garbage can be collected at one node without communicating with other nodes. Garbage collection can be performed concurrently with the application and imposes no message ordering constraints. We also expect DRL to be reasonably efficient in practice, since it does not require many additional messages or significant actor-local computation.  

DRL works as follows.  The \emph{communication protocol} (\cref{sec:model}) tracks information, such as references and message counts, and stores it in each actor's state.  Actors periodically send out copies of their local state (called \emph{snapshots}) to be stored at one or more designated \emph{snapshot aggregator} actors. Each aggregator periodically searches its local store to find a subset of snapshots representing terminated actors (\cref{sec:termination-detection}).  Once an actor is determined to have terminated, it can be garbage collected by, for example, sending it a \emph{self-destruct} message.  Note that our termination detection algorithm itself is \textit{location transparent}. 

Since DRL is defined on top of the actor model, it is oblivious to details of a particular implementation (such as how sequential computations are represented). Our technique is therefore applicable to any actor framework and can be
implemented as a library.  Moreover, it can also be applied to open
systems, allowing a garbage-collected actor subsystem to interoperate with an external actor
system.

The outline of the paper is as follows. We provide a characterization of actor garbage in Section~\ref{sec:background} and discuss related work in Section~\ref{sec:related-work}. We then provide a
specification of the DRL protocol in Section~\ref{sec:model}.
In Section~\ref{sec:chain-lemma}, we describe a key property of DRL called the \emph{Chain Lemma}. This lemma allows us to prove the safety and liveness properties, which are stated in
Section~\ref{sec:termination-detection}. We then conclude in Section~\ref{sec:future-work} with some discussion of future work and how DRL may be used in practice. To conserve space, all proofs have been relegated to the Appendix.

\section{Preliminaries}
\label{sec:background}

An actor can only receive a message when it is \emph{idle}. Upon
receiving a message, it becomes \emph{busy}. A busy actor can perform
an unbounded sequence of \emph{actions} before becoming idle. In~\cite{aghaFoundationActorComputation1997}, an action may be to spawn an actor, send a message, or perform a (local) computation. We will also assume that actors can perform effects, such as file I/O. The actions an actor performs in response to a message are dictated by its application-level code, called a \emph{behavior}.

Actors can also receive messages from \emph{external} actors (such as the user) by
becoming \emph{receptionists}. An actor $A$ becomes a receptionist when its address is exposed to an external actor. Subsequently, any external actor can potentially obtain $A$'s address and send it a message. It is not possible for an actor system to determine when all external actors have ``forgotten'' a receptionist's address. We will therefore assume that an actor can never cease to be a receptionist once its address has been exposed.

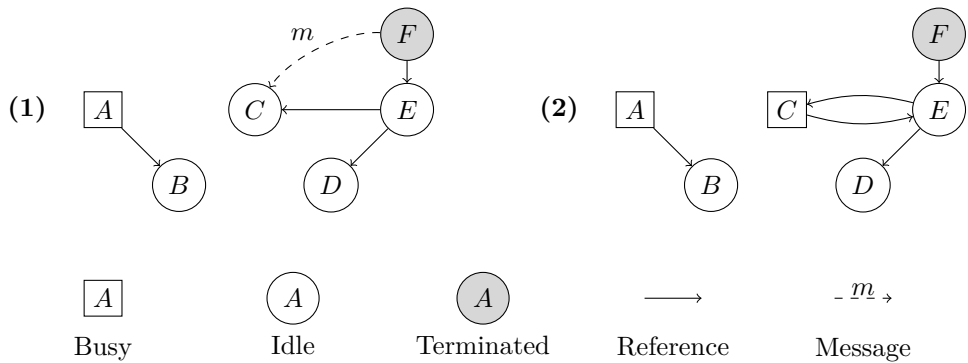
\begin{figure} \centering \begin{tikzpicture}
	\begin{pgfonlayer}{nodelayer}
		\node [style=none] (6) at (-4.75, -2) {$m$};
		\node [style=actor] (51) at (-8, -6) {$B$};
		\node [style=busyActor] (53) at (-10, -4) {$A$};
		\node [style=actor] (54) at (-6, -4) {$C$};
		\node [style=terminatedActor] (55) at (-2, -2) {$F$};
		\node [style=actor] (56) at (-4, -6) {$D$};
		\node [style=actor] (57) at (-2, -4) {$E$};
		\node [style=busyActor] (58) at (-10, -9) {$A$};
		\node [style=none] (59) at (-10, -10.325) {Busy};
		\node [style=none] (60) at (4.25, -9) {};
		\node [style=none] (61) at (5.75, -9) {};
		\node [style=none] (62) at (9.25, -9) {};
		\node [style=none] (63) at (10.75, -9) {};
		\node [style=none] (64) at (5, -10.25) {};
		\node [style=none] (65) at (5, -10.275) {Reference};
		\node [style=none] (66) at (10, -10.35) {Message};
		\node [style=none] (68) at (10, -8.75) {$m$};
		\node [style=actor] (69) at (-5, -9) {$A$};
		\node [style=none] (70) at (-5, -10.25) {Idle};
		\node [style=actor] (72) at (6, -6) {$B$};
		\node [style=busyActor] (73) at (4, -4) {$A$};
		\node [style=busyActor] (74) at (8, -4) {$C$};
		\node [style=terminatedActor] (75) at (12, -2) {$F$};
		\node [style=actor] (76) at (10, -6) {$D$};
		\node [style=actor] (77) at (12, -4) {$E$};
		\node [style=none] (78) at (-12, -4) {\textbf{(1)}};
		\node [style=none] (79) at (2, -4) {\textbf{(2)}};
		\node [style=terminatedActor] (80) at (0, -9) {$A$};
		\node [style=none, align=center] (81) at (0, -10.25) {Terminated};
	\end{pgfonlayer}
	\begin{pgfonlayer}{edgelayer}
		\draw [style=reference] (53) to (51);
		\draw [style=reference] (57) to (56);
		\draw [style=reference] (57) to (54);
		\draw [style=message, bend left=330, looseness=0.75] (55) to (54);
		\draw [style=reference] (55) to (57);
		\draw [style=reference] (60.center) to (61.center);
		\draw [style=message] (62.center) to (63.center);
		\draw [style=reference] (73) to (72);
		\draw [style=reference] (77) to (76);
		\draw [style=reference, bend left=345] (77) to (74);
		\draw [style=reference] (75) to (77);
		\draw [style=reference, bend right=15] (74) to (77);
	\end{pgfonlayer}
\end{tikzpicture}
\caption{A simple actor system. The first configuration leads to the second after $C$ receives the message $m$, which contains a reference to $E$. Notice that an actor can send a message and ``forget'' its reference to the recipient before the message is delivered, as is the case for actor $F$. In both configurations, $E$ is a potential acquaintance of $C$, and $D$ is potentially reachable from $C$. The only terminated actor is $F$ because all other actors are potentially reachable from unblocked actors.}
\label{fig:actor-graph-example}
\end{figure}

An actor is said to be garbage if it can be destroyed without affecting the system's observable behavior. However, without analyzing an actor’s code, it is not possible to know whether it will have an effect when it receives a message. We will therefore restrict our attention to actors that can be guaranteed to be garbage without inspecting their behavior. According to this more conservative definition, any actor that might receive a message in the future should not be garbage collected because it could, for instance, write to a log file when it becomes busy. Conversely, any actor that is guaranteed to remain idle indefinitely can safely be garbage collected because it will never have any effects; such an actor is said to be \emph{terminated}.
Hence, garbage actors coincide with terminated actors in our model. %This characterization of actor garbage is also used in~\cite{clebschFullyConcurrentGarbage2013,wangConservativeSnapshotbasedActor2011,wangDistributedGarbageCollection2006}.

Terminated actors can be detected by looking at the global state of the system. We say that an actor $B$ is a \emph{potential acquaintance} of $A$ (and $A$ is a \emph{potential inverse acquaintance} of $B$) if $A$ has a reference to $B$ or if there is an undelivered message to $A$ that contains a reference to $B$. We define \emph{potential reachability} to be the reflexive transitive closure of the potential acquaintance relation. If an actor is idle and has no undelivered messages, then it is \emph{blocked}; otherwise it is \emph{unblocked}. We then observe that an actor is terminated when it is only potentially reachable by blocked actors: Such an actor is idle, blocked, and can only potentially be sent a message by other idle blocked actors. Conversely, without analyzing actor code we cannot safely conclude that an actor is terminated if it is potentially reachable by an unblocked actor. Hence, we say that an actor is terminated if and only if it is blocked and all of its potential inverse acquaintances are terminated.

\section{Related Work}\label{sec:related-work}

\paragraph*{Global Termination}  
\emph{Global} termination detection (GTD) is used to determine when
\emph{all} processes have terminated
\cite{matternAlgorithmsDistributedTermination1987,matochaTaxonomyDistributedTermination1998}.
For GTD, it suffices to obtain global message send and receive counts.
Most GTD algorithms also assume a fixed process topology. However, Lai
gives an algorithm in \cite{laiTerminationDetectionDynamically1986}
that supports dynamic topologies such as in the actor model. Lai's
algorithm performs termination detection in ``waves'', disseminating
control messages along a spanning tree (such as an actor supervisor
hierarchy) so as to obtain consistent global message send and receive
counts. Venkatasubramanian et al.~take a similar approach to obtain a
consistent global snapshot of actor states in a distributed
system~\cite{venkatasubramanianScalableDistributedGarbage1992}. However,
such an approach does not scale well because it is not incremental:
garbage cannot be detected until all nodes in the system have
responded.  In contrast, DRL does not require a global snapshot, does not
require actors to coordinate their local snapshots, and does not
require waiting for all nodes before detecting local terminated
actors.

\paragraph*{Reference Tracking} We say that an idle actor is \emph{simple garbage} if it has no undelivered messages and no other actor has a reference to it.
Such actors can be detected with distributed reference counting
\cite{watsonEfficientGarbageCollection1987,bevanDistributedGarbageCollection1987,piquerIndirectReferenceCounting1991}
or with reference listing
\cite{DBLP:conf/iwmm/PlainfosseS95,wangDistributedGarbageCollection2006}
techniques.  In reference listing algorithms, each actor maintains a
partial list of actors that may have references to it. Whenever $A$ sends $B$ a
reference to $C$, it also sends an $\InfoMsg$ message informing $C$
about $B$'s reference. Once $B$ no longer needs a reference to $C$, it
informs $C$ by sending a $\ReleaseMsg$ message; this message should not be processed by $C$ until all preceding messages from $B$ to $C$ have been delivered. Thus an actor is
simple garbage when its reference listing is
empty. 

Our technique uses a form of \emph{deferred reference listing}, in which $A$ may also defer sending $\InfoMsg$
messages to $C$ until it releases its references to $C$.  This allows
$\InfoMsg$ and $\ReleaseMsg$ messages to be batched together, reducing communication
overhead.

\paragraph*{Cyclic Garbage}

Actors that are transitively acquainted with one another are said to
form cycles. Cycles of terminated actors are called \emph{cyclic
garbage} and cannot be detected with reference listing alone.  Since
actors are hosted on nodes and cycles may span across multiple nodes,
detecting cyclic garbage requires sharing information between nodes to
obtain a consistent view of the global topology.  One approach is to
compute a global snapshot of the distributed system
\cite{kafuraConcurrentDistributedGarbage1995} using the Chandy-Lamport
algorithm \cite{chandyDistributedSnapshotsDetermining1985}; but this
requires pausing execution of all actors on a node to compute its
local snapshot.

Another approach is to add edges to the actor reference graph so
that actor garbage coincides with passive object garbage
\cite{vardhanUsingPassiveObject2003,wangActorGarbageCollection2010}. This
is convenient because it allows existing algorithms for distributed
passive object GC, such as
\cite{schelvisIncrementalDistributionTimestamp1989}, to be reused in
actor systems. However, such transformations require that actors know
when they have undelivered messages, which requires some form of
synchronization.

To avoid pausing executions, Wang and Varela proposed a reference
listing based technique called the \emph{pseudo-root} algorithm.  The
algorithm computes \emph{approximate} global snapshots and is
implemented in the SALSA runtime
\cite{wangDistributedGarbageCollection2006,wangConservativeSnapshotbasedActor2011}.
The pseudo-root algorithm requires a high number of additional control
messages and requires actors to write to shared memory if they migrate
or release references during snapshot collection.  Our protocol
requires fewer control messages and no additional actions between
local actor snapshots.  Wang and Varela also explicitly address migration of actors, 
a concern orthogonal to our algorithm.

Our technique is inspired by \emph{MAC}, a termination detection
algorithm implemented in the Pony runtime
\cite{clebschFullyConcurrentGarbage2013}. In MAC, actors send a local
snapshot to a designated cycle detector whenever their message queue
becomes empty, and send another notification whenever it becomes non-empty. Clebsch and Drossopoulou prove that for systems with
causal message delivery, a simple request-reply protocol is sufficient
to confirm that the cycle detector's view of the topology is
consistent.  However, enforcing causal delivery in a distributed
system imposes additional space and networking costs
\cite{fidge1987timestamps,blessingTreeTopologiesCausal2017}. DRL is
similar to MAC, but does not require causal message delivery, supports
decentralized termination detection, and actors need not take
snapshots each time their message queues become empty. The key insight is that these limitations can be removed by tracking additional information at the actor level.

An earlier version of DRL appeared in
\cite{plyukhinConcurrentGarbageCollection2018}. In this paper, we
formalize the description of the algorithm and prove its safety and
liveness.  In the process, we discovered that release acknowledgment
messages are unnecessary and that termination detection is more
flexible than we first thought: it is not necessary for GC to be
performed in distinct ``phases'' where every actor takes a snapshot in
each phase.  In particular, once an idle actor takes a snapshot, it
need not take another snapshot until it receives a fresh message.

\section{A Two-Level Semantic Model}\label{sec:model}

Our computation model is based on the two
level approach to actor semantics
\cite{venkatasubramanianReasoningMetaLevel1995}, in which a lower \emph{system-level} transition system interprets the operations performed by a higher, user-facing
\emph{application-level} transition system. In this section, we define the DRL communication protocol at the system level. We do not provide a
transition system for the application level computation model, since it is
not relevant to garbage collection (see
\cite{aghaFoundationActorComputation1997} for how it can be
done).  What is relevant to us is that corresponding to each
application-level action is a system-level transition that tracks
references.
We will therefore define \emph{system-level configurations} and
\emph{transitions on system-level configurations}.  We will refer to
these, respectively, as configurations and transitions in the rest of
the paper.

%Every actor has system-level local state and
%can send and receive system-level messages which are independent of
%the application level.  Because system level messages do not interfere
%with the application level messages, GC is transparent to the
%application.  

% In our transition system, we ignore the structure of application
% level computations.  A similar approach is used to define the MAC
% actor garbage collection algorithm
% \cite{clebschFullyConcurrentGarbage2013}.

\subsection{Overview}
\label{sec:overview}

Actors in DRL use \emph{reference objects} (abbreviated \emph{refobs}) to send messages, instead of using plain actor addresses.  Refobs are similar to unidirectional channels and can only be used by their designated \emph{owner} to send messages to their \emph{target}; thus in order for $A$ to give $B$ a reference to $C$, it must explicitly create a new refob owned by $B$. Once a refob is no longer needed, it should be \emph{deactivated} by its owner and removed from local state.

The DRL communication protocol enriches each actor's state with a list of refobs that it currently owns and associated message counts representing the number of messages sent using each refob. Each actor also maintains a subset of the refobs of which it is the target, together with associated message receive counts. Lastly, actors perform a form of ``contact tracing'' by maintaining a subset of the refobs that they have created for other actors; we provide details about the bookkeeping later in this section.

The additional information above allows us to detect termination by inspecting actor snapshots. If a set of snapshots is consistent (in the sense of \cite{chandyDistributedSnapshotsDetermining1985}) then we can use the ``contact tracing'' information to determine whether the set is \emph{closed} under the potential inverse acquaintance relation (see \cref{sec:chain-lemma}). Then, given a consistent and closed set of snapshots, we can use the message counts to determine whether an actor is blocked. We can therefore find all the terminated actors within a consistent set of snapshots.

In fact, DRL satisfies a stronger property: any set of snapshots that ``appears terminated'' in the sense above is guaranteed to be consistent. Hence, given an arbitrary  closed set of snapshots, it is possible to determine which of the corresponding actors have terminated. This allows a great deal of freedom in how snapshots are aggregated. For instance, actors could place their snapshots in a global eventually consistent store, with a garbage collection thread at each node periodically inspecting the store for local terminated actors.

\paragraph*{Reference Objects}

\begin{figure} \centering \begin{tikzpicture}
	\begin{pgfonlayer}{nodelayer}
		\node [style=actor] (0) at (-11.5, 1) {$C$};
		\node [style=busyActor] (1) at (-6.5, 1) {$A$};
		\node [style=none] (6) at (-9, 1.25) {$y$};
		\node [style=actor] (7) at (3.5, 1) {$C$};
		\node [style=busyActor] (8) at (8.5, 1) {$A$};
		\node [style=actor] (9) at (13.5, 1) {$B$};
		\node [style=none] (10) at (11, 1.25) {$x$};
		\node [style=none] (11) at (6, 1.25) {$y$};
		\node [style=actor] (12) at (-11.5, -2.5) {$C$};
		\node [style=busyActor] (13) at (-6.5, -2.5) {$A$};
		\node [style=actor] (14) at (-1.5, -2.5) {$B$};
		\node [style=none] (15) at (-4, -2.25) {$x$};
		\node [style=none] (16) at (-9, -2.25) {$y$};
		\node [style=none] (17) at (-6.5, -3.75) {$\mathtt{CreatedUsing}(y,z)$};
		\node [style=none] (18) at (-4, -1) {$\mathtt{app}(x,\{z\})$};
		\node [style=actor] (19) at (3.5, -2.5) {$C$};
		\node [style=busyActor] (20) at (8.5, -2.5) {$A$};
		\node [style=busyActor] (21) at (13.5, -2.5) {$B$};
		\node [style=none] (22) at (11.5, -2.25) {$x$};
		\node [style=none] (23) at (6, -2.25) {$y$};
		\node [style=none] (24) at (13.5, -3.75) {$\mathtt{Active}(z)$};
		\node [style=none] (25) at (6, -1) {$\mathtt{info}(y,z)$};
		\node [style=actor] (26) at (-11.5, -6) {$C$};
		\node [style=actor] (27) at (-6.5, -6) {$A$};
		\node [style=busyActor] (28) at (-1.5, -6) {$B$};
		\node [style=none] (29) at (-4, -5.75) {$x$};
		\node [style=none] (30) at (-9, -5.75) {$y$};
		\node [style=none] (31) at (-11.5, -7.25) {$\mathtt{Created}(z)$};
		\node [style=none] (32) at (-6.5, -8) {$\mathtt{release}(z)$};
		\node [style=actor] (33) at (3.5, -6) {$C$};
		\node [style=actor] (34) at (8.5, -6) {$A$};
		\node [style=actor] (35) at (13.5, -6) {$B$};
		\node [style=none] (36) at (11, -5.75) {$x$};
		\node [style=none] (37) at (6, -5.75) {$y$};
		\node [style=none] (38) at (3.5, -7.25) {$\mathtt{Created}(z)$};
		\node [style=none] (39) at (3.5, -8) {$\mathtt{Released}(z)$};
		\node [style=busyActor] (40) at (-2.5, -10) {$A$};
		\node [style=none] (42) at (1.75, -10) {};
		\node [style=none] (43) at (3.25, -10) {};
		\node [style=none] (44) at (6.75, -10) {};
		\node [style=none] (45) at (8.25, -10) {};
		\node [style=none] (49) at (2.5, -9.75) {$x$};
		\node [style=none] (50) at (7.5, -9.75) {$m$};
		\node [style=none] (51) at (1.5, 1) {\textbf{(2)}};
		\node [style=none] (52) at (1.5, -2.5) {\textbf{(4)}};
		\node [style=none] (53) at (1.5, -6) {\textbf{(6)}};
		\node [style=none] (54) at (-13.5, 1) {\textbf{(1)}};
		\node [style=none] (55) at (-13.5, -2.5) {\textbf{(3)}};
		\node [style=none] (56) at (-13.5, -6) {\textbf{(5)}};
		\node [style=none] (57) at (8.5, -4.25) {$z$};
		\node [style=actor] (58) at (-7.5, -10) {$A$};
		\node [style=none] (60) at (-2.5, -11.325) {Busy Actor};
		\node [style=none] (61) at (2.5, -11.25) {};
		\node [style=none] (62) at (2.5, -11.275) {Reference};
		\node [style=none] (63) at (7.5, -11.35) {Message};
		\node [style=none] (64) at (-7.5, -11.25) {Idle Actor};
	\end{pgfonlayer}
	\begin{pgfonlayer}{edgelayer}
		\draw [style=reference] (1) to (0);
		\draw [style=reference, in=180, out=0] (8) to (9);
		\draw [style=reference] (8) to (7);
		\draw [style=reference, in=180, out=0] (13) to (14);
		\draw [style=message, in=150, out=30] (13) to (14);
		\draw [style=reference] (13) to (12);
		\draw [style=reference, in=180, out=0] (20) to (21);
		\draw [style=reference] (20) to (19);
		\draw [style=message, bend left=330] (20) to (19);
		\draw [style=reference, in=180, out=0] (27) to (28);
		\draw [style=reference] (27) to (26);
		\draw [style=message, bend left, looseness=0.75] (28) to (26);
		\draw [style=reference, in=180, out=0] (34) to (35);
		\draw [style=reference] (34) to (33);
		\draw [style=reference] (42.center) to (43.center);
		\draw [style=message] (44.center) to (45.center);
		\draw [style=reference, bend left, looseness=0.75] (21) to (19);
	\end{pgfonlayer}
\end{tikzpicture}
    \caption{An example showing how refobs are created and
destroyed. Below each actor we list all the ``facts'' related to $z$ that are stored in its local state. Although not pictured in the figure, $A$ also obtains facts $\Activated(x)$ and $\Activated(y)$ after spawning actors $B$ and $C$, respectively. Likewise, actors $B,C$ obtain facts $\Created(x),\Created(y)$, respectively, upon being spawned.}
    \label{fig:refob-example}
\end{figure}
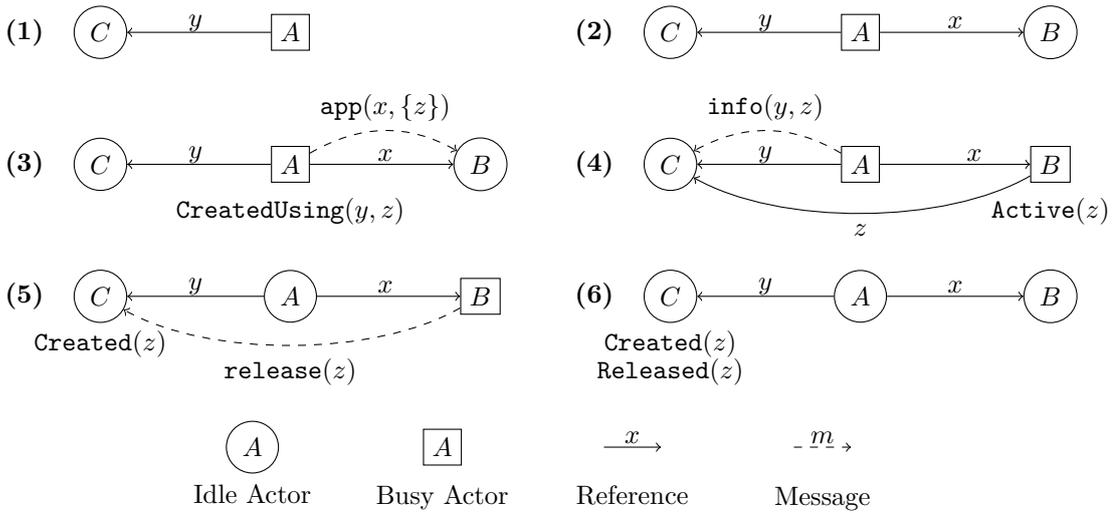

A refob is a triple $(x,A,B)$, where $A$ is the owner actor's address, $B$ is the target actor's address, and $x$ is a globally unique token. An actor can cheaply generate such a token by combining its address with a local sequence number, since actor systems already guarantee that each address is unique. We will stylize a triple $(x,A,B)$ as $\Refob x A B$. We will also sometimes refer to such a refob as simply $x$, since tokens act as unique identifiers.

When an actor $A$ spawns an actor $B$ (Fig.~\ref{fig:refob-example}
(1, 2)) the DRL protocol creates a new refob
$\Refob x A B$ that is stored in both $A$ and $B$'s system-level
state, and a refob $\Refob y B B$ in $B$'s state. The refob $x$ allows $A$ to send application-level messages to
$B$. These messages are denoted $\AppMsg(x,R)$, where $R$ is the sett of refobs contained in the message that $A$ has created for $B$. The refob $y$ corresponds to the \texttt{self} variable present in some actor languages. 

If $A$ has active refobs $\Refob x A B$ and $\Refob y A C$, then it can
create a new refob $\Refob z B C$ by generating a token $z$. In
addition to being sent to $B$, this refob must also temporarily be
stored in $A$'s system-level state and marked as ``created using $y$''
(Fig.~\ref{fig:refob-example} (3)). Once $B$ receives $z$, it must add
the refob to its system-level state and mark it as ``active''
(Fig.~\ref{fig:refob-example} (4)). Note that $B$ can have multiple
distinct refobs that reference the same actor in its state; this
can be the result of, for example, several actors concurrently sending
refobs to $B$.  Transition rules for spawning actors and sending
messages are given in Section~\ref{sec:standard-actor-operations}.

Actor $A$ may remove $z$ from its state once it has sent a
(system-level) $\InfoMsg$ message informing $C$ about $z$
(Fig.~\ref{fig:refob-example} (4)). Similarly, when $B$ no longer
needs its refob for $C$, it can ``deactivate'' $z$ by removing it
from local state and sending $C$ a (system-level) $\ReleaseMsg$
message (Fig.~\ref{fig:refob-example} (5)). Note that if $B$ already
has a refob $\Refob z B C$ and then receives another $\Refob {z'} B C$,
then it can be more efficient to defer deactivating the extraneous
$z'$ until $z$ is also no longer needed; this way, the $\ReleaseMsg$
messages can be batched together.

When $C$ receives an $\InfoMsg$ message, it records that the refob
has been created, and when $C$ receives a $\ReleaseMsg$ message, it
records that the refob has been released
(Fig.~\ref{fig:refob-example} (6)).  Note that these messages may
arrive in any order. Once $C$ has received both, it is permitted to
remove all facts about the refob from its local state. Transition
rules for these reference listing actions are given in
Section~\ref{sec:release-protocol}.

Once a refob has been created, it cycles through four states:
pending, active, inactive, or released.  A refob $\Refob z B C$ is
said to be \emph{pending} until it is received by its owner $B$. Once
received, the refob is \emph{active} until it is \emph{deactivated}
by its owner, at which point it becomes \emph{inactive}.  Finally,
once $C$ learns that $z$ has been deactivated, the refob is said to
be \emph{released}.  A refob that has not yet been released is
\emph{unreleased}.  

Slightly amending the definition we gave in \cref{sec:background}, we say that $B$ is a \emph{potential acquaintance} of $A$
(and $A$ is a \emph{potential inverse acquaintance} of $B$) when there
exists an unreleased refob $\Refob x A B$. Thus, $B$ becomes a potential acquaintance of $A$ as soon as $x$ is created, and only ceases to be an acquaintance once it has received a $\ReleaseMsg$ message for every refob $\Refob y A B$ that has been created so far.

\begin{figure} \centering \begin{tikzpicture}
	\begin{pgfonlayer}{nodelayer}
		\node [style=none] (98) at (-12.75, 5) {};
		\node [style=none] (99) at (-12.75, 5) {$A$};
		\node [style=none] (100) at (-12, 5) {};
		\node [style=none] (101) at (12, 5) {};
		\node [style=none] (102) at (-12, 0) {};
		\node [style=none] (103) at (12, 0) {};
		\node [style=none] (104) at (-12, -5) {};
		\node [style=none] (105) at (12, -5) {};
		\node [style=none] (106) at (-12.75, 0) {$B$};
		\node [style=none] (107) at (-12.75, -5) {$C$};
		\node [style=none] (114) at (1.75, 5) {};
		\node [style=none] (115) at (4.5, 0) {};
		\node [style=none] (116) at (-4.5, 0.25) {};
		\node [style=none] (117) at (-2.5, -5.75) {};
		\node [style=none] (122) at (-3, -2.5) {$t_1$};
		\node [style=none] (133) at (-10.5, 0.25) {};
		\node [style=none] (134) at (-8.5, -5.75) {};
		\node [style=none] (135) at (-9, -2.5) {$t_0$};
		\node [style=none] (140) at (-4.5, 5) {};
		\node [style=none] (141) at (-1.75, 0) {};
		\node [style=none] (143) at (-11, 5) {};
		\node [style=none] (144) at (-8.25, 0) {};
		\node [style=none, rotate=-59] (145) at (-8.975, 2.45) {\small $\mathtt{app}(x)$};
		\node [style=none] (146) at (-7.75, 0) {};
		\node [style=none] (147) at (-5, -5) {};
		\node [style=none, rotate=-59] (148) at (-5.725, -2.55) {\small $\mathtt{app}(y)$};
		\node [style=none] (149) at (-0.75, 0) {};
		\node [style=none] (150) at (2, -5) {};
		\node [style=none, rotate=-59] (151) at (1.275, -2.55) {\small $\mathtt{release}(y)$};
		\node [style=none] (152) at (5.5, 0) {};
		\node [style=none] (153) at (8.25, -5) {};
		\node [style=none, rotate=-59] (154) at (7.525, -2.55) {\small $\mathtt{app}(y')$};
		\node [style=none] (155) at (2.25, 0.25) {};
		\node [style=none] (156) at (4.25, -5.75) {};
		\node [style=none] (157) at (3.75, -2.5) {$t_2$};
		\node [style=none] (158) at (8.5, 0.25) {};
		\node [style=none] (159) at (10.5, -5.75) {};
		\node [style=none] (160) at (10, -2.5) {$t_3$};
		\node [style=none, rotate=-59] (161) at (-2.5, 2.45) {\small $\mathtt{app}(x)$};
		\node [style=none, rotate=-59] (162) at (4, 2.45) {\small $\mathtt{app}(x,\{y'\})$};
		\node [style=none] (163) at (-10.5, 0.75) {\small $\mathtt{Sent}(y,1)$};
		\node [style=none] (164) at (-4.5, 0.75) {\small $\mathtt{Sent}(y,2)$};
		\node [style=none] (165) at (8.5, 0.75) {\small $\mathtt{Sent}(y',1)$};
		\node [style=none] (166) at (2.25, 0.75) {\small $\emptyset$};
		\node [style=none] (167) at (-8.5, -6.25) {\small $\mathtt{Recv}(y,1)$};
		\node [style=none] (168) at (-2.5, -6.25) {\small $\mathtt{Recv}(y,2)$};
		\node [style=none] (169) at (10.5, -6.25) {\small $\mathtt{Recv}(y',1)$};
		\node [style=none] (170) at (4.25, -6.25) {\small $\emptyset$};
	\end{pgfonlayer}
	\begin{pgfonlayer}{edgelayer}
		\draw [style=reference] (100.center) to (101.center);
		\draw [style=reference] (102.center) to (103.center);
		\draw [style=reference] (104.center) to (105.center);
		\draw [style=message] (114.center) to (115.center);
		\draw [style=time horizon, in=90, out=-90, looseness=1.25] (116.center) to (117.center);
		\draw [style=time horizon, in=90, out=-90, looseness=1.25] (133.center) to (134.center);
		\draw [style=message] (140.center) to (141.center);
		\draw [style=message] (143.center) to (144.center);
		\draw [style=message] (146.center) to (147.center);
		\draw [style=message] (149.center) to (150.center);
		\draw [style=message] (152.center) to (153.center);
		\draw [style=time horizon, in=90, out=-90, looseness=1.25] (155.center) to (156.center);
		\draw [style=time horizon, in=90, out=-90, looseness=1.25] (158.center) to (159.center);
	\end{pgfonlayer}
\end{tikzpicture}
    \caption{A time diagram for actors $A,B,C$, demonstrating message counts and consistent snapshots. Dashed arrows represent messages and dotted lines represent consistent cuts. In each cut above, $B$'s message send count agrees with $C$'s message receive count.}
    \label{fig:message-counts}
\end{figure}
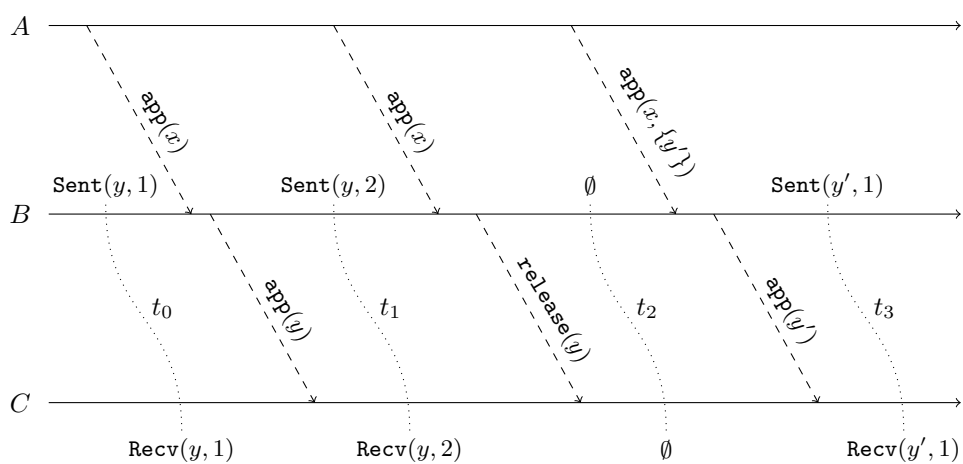

\paragraph*{Message Counts and Snapshots}

For each refob $\Refob x A B$, the owner $A$ counts the
number of $\AppMsg$ and $\InfoMsg$ messages sent along $x$; this count can be deleted when $A$
deactivates $x$. Each message is annotated with the refob used to
send it. Whenever $B$ receives an $\AppMsg$ or $\InfoMsg$ message along $x$, it
correspondingly increments a receive count for $x$; this count can be deleted once $x$
has been released. Thus the memory overhead of message counts is linear in
the number of unreleased refobs.

A snapshot is a copy of all the facts in an actor's system-level state at some point in time. We will assume throughout the paper that in every set of snapshots $Q$, each snapshot was taken by a different actor. Such a set is also said to form a \emph{cut}. Recall that a cut is consistent if no snapshot in the cut causally precedes any other \cite{chandyDistributedSnapshotsDetermining1985}. Let us also say that $Q$ is a set of \emph{mutually quiescent} snapshots if there are no undelivered messages between actors in the cut. That is, if $A \in Q$ sent a message to $B \in Q$ before taking a snapshot, then the message must have been delivered before $B$ took its snapshot. Notice that if all snapshots in $Q$ are mutually quiescent, then $Q$ is consistent.

Notice also that in Fig.~\ref{fig:message-counts}, the snapshots of $B$ and $C$ are mutually quiescent when their send and receive counts agree. This is ensured in part because each refob has a unique token: If actors associated message counts with actor names instead of tokens, then $B$’s snapshots at $t_0$ and $t_3$ would both contain $\SentCount(C,1)$. Thus, $B$’s snapshot at $t_3$ and $C$’s snapshot at $t_0$ would appear mutually quiescent, despite having undelivered messages in the cut.

We would like to conclude that snapshots from two actors $A,B$ are mutually quiescent if and only if their send and receive counts are agreed for every refob $\Refob x A B$ or $\Refob y B A$. Unfortunately, this fails to hold in general for systems with unordered message delivery.  It also fails to hold when, for instance, the owner actor takes a snapshot before the refob is activated and the target actor takes a snapshot after the refob is released. In such a case, neither knowledge set includes a message count for the refob and they therefore appear to agree.  However, we show that the message counts can nevertheless be used to bound the number of undelivered messages for purposes of our algorithm (\cref{lem:msg-counts}).

\paragraph*{Definitions}

We use the capital letters $A,B,C,D,E$ to denote actor addresses.
Tokens are denoted $x,y,z$, with a special reserved token $\NullToken$
for messages from external actors.

A \emph{fact} is a value that takes one of the following forms:
$\Created(x)$, $\Released(x)$, $\CreatedUsing(x,y)$, $\Activated(x)$, $\Unreleased(x)$,
$\SentCount(x,n)$, or $\RecvCount(x,n)$ for some refobs $x,y$ and
natural number $n$.  Each actor's state holds a set of facts about
refobs and message counts called its \emph{knowledge set}.  We use
$\phi,\psi$ to denote facts and $\Phi,\Psi$ to denote finite sets of
facts.  Each fact may be interpreted as a \emph{predicate} that
indicates the occurrence of some past event. Interpreting a set of
facts $\Phi$ as a set of axioms, we write $\Phi \vdash \phi$ when
$\phi$ is derivable by first-order logic from $\Phi$ with the
following additional rules: 
\begin{itemize}
    \item If $(\not\exists n \in \mathbb N,\ \SentCount(x,n) \in
\Phi)$ then $\Phi \vdash \SentCount(x,0)$
    \item If $(\not\exists n \in \mathbb N,\ \RecvCount(x,n) \in
\Phi)$ then $\Phi \vdash \RecvCount(x,0)$
    \item If $\Phi \vdash \Created(x) \land \lnot \Released(x)$ then
$\Phi \vdash \Unreleased(x)$
    \item If $\Phi \vdash \CreatedUsing(x,y)$ then $\Phi \vdash
\Created(y)$
\end{itemize}

For convenience, we define a pair of functions
$\IncSent(x,\Phi),\IncRecv(x,\Phi)$ for incrementing message
send/receive counts, as follows: If $\SentCount(x,n) \in \Phi$ for some
$n$, then
$\IncSent(x,\Phi) = (\Phi \setminus \{\SentCount(x,n)\}) \cup
\{\SentCount(x,n+1)\}$; otherwise,
$\IncSent(x,\Phi) = \Phi \cup \{\SentCount(x,1)\}$. Likewise for
$\IncRecv$ and $\RecvCount$.

Recall that an actor is either \emph{busy} (processing a message) or
\emph{idle} (waiting for a message). An actor with knowledge set
$\Phi$ is denoted $[\Phi]$ if it is busy and $(\Phi)$ if it is idle.

Our specification includes both \emph{system messages} (also called
\emph{control messages}) and \emph{application messages}. The former
are automatically generated by the DRL  protocol and handled at the system
level, whereas the latter are explicitly created and consumed by
user-defined behaviors. Application-level messages are denoted
$\AppMsg(x,R)$. The argument $x$ is the refob used to send the
message. The second argument $R$ is a set of refobs created by the
sender to be used by the destination actor. Any remaining application-specific data in the message is omitted in our notation.

The DRL communication protocol uses two kinds of system messages. $\InfoMsg(y, z, B)$ is a message sent from an actor $A$ to an actor $C$, informing it that a new refob $\Refob z B C$ was created using $\Refob y A C$. $\ReleaseMsg(x,n)$ is a message sent from an actor $A$ to an actor $B$, informing it that the refob $\Refob x A B$ has been deactivated and should be released.

A \emph{configuration} $\Config{\alpha}{\mu}{\rho}{\chi}$ is a
quadruple $(\alpha,\mu,\rho,\chi)$ where: $\alpha$ is a mapping from actor addresses to knowledge sets; $\mu$ is a mapping from actor addresses to multisets of messages; and $\rho,\chi$ are sets of actor addresses. Actors in $\dom(\alpha)$ are \emph{internal actors} and actors in $\chi$ are
\emph{external actors}; the two sets may not intersect. The mapping $\mu$ associates each actor with undelivered messages to that actor. Actors in
$\rho$ are \emph{receptionists}.  We will ensure $\rho \subseteq \dom(\alpha)$ remains
valid in any configuration that is derived from a configuration where
the property holds (referred to as the locality laws in
\cite{Baker-Hewitt-laws77}).

Configurations are denoted by $\kappa$, $\kappa'$, $\kappa_0$,
etc. If an actor address $A$ (resp. a token $x$), does not occur in
$\kappa$, then the address (resp. the token) is said to be
\emph{fresh}.  We assume a facility for generating fresh addresses and
tokens.

In order to express our transition rules in a pattern-matching style, we will employ the following shorthand. Let $\alpha,[\Phi]_A$ refer to a
mapping $\alpha'$ where $\alpha'(A) = [\Phi]$ and $\alpha =
\alpha'|_{\dom(\alpha') \setminus \{A\}}$.  Similarly, let
$\mu,\Msg{A}{m}$ refer to a mapping $\mu'$ where $m \in \mu'(A)$ and
$\mu = \mu'|_{\dom(\mu') \setminus \{A\}} \cup \{A \mapsto \mu'(A)
\setminus \{m\}\}$. Informally, the expression $\alpha,[\Phi]_A$ refers to a set of actors containing both $\alpha$ and the busy actor $A$ (with knowledge set $\Phi$); the expression $\mu, \Msg{A}{m}$ refers to the set of messages containing both $\mu$ and the message $m$ (sent to actor $A$).

The rules of our transition system define atomic transitions from one configuration
to another.  Each transition rule has a label $l$, parameterized by some
variables $\vec x$ that occur in the left- and right-hand
configurations. Given a configuration $\kappa$, these parameters
functionally determine the next configuration $\kappa'$. Given
arguments $\vec v$, we write $\kappa \Step{l(\vec v)} \kappa'$ to denote a semantic step from $\kappa$ to $\kappa'$ using rule $l(\vec v)$.

We refer to a label with arguments $l(\vec v)$ as an \emph{event},
denoted $e$. A sequence of events is denoted $\pi$. If $\pi =
e_1,\dots,e_n$ then we write $\kappa \Step \pi \kappa'$ when $\kappa
\Step{e_1} \kappa_1 \Step{e_2} \dots \Step{e_n} \kappa'$. If there
exists $\pi$ such that $\kappa \Step \pi \kappa'$, then $\kappa'$ is
\emph{derivable} from $\kappa$. An \emph{execution} is a sequence of events $e_1,\dots,e_n$ such that
$\kappa_0 \Step{e_1} \kappa_1 \Step{e_2} \dots \Step{e_n} \kappa_n$,
where $\kappa_0$ is the initial configuration
(Section~\ref{sec:initial-configuration}).  We say that a property holds \emph{at time $t$} if it holds in $\kappa_t$.

\subsection{Initial Configuration}\label{sec:initial-configuration}

The initial configuration $\kappa_0$ consists of a single actor in a
busy state:
$$\Config{[\Phi]_A}{\emptyset}{\emptyset}{\{E\}},$$
where
$\Phi = \{\Activated(\Refob x A E),\ \Created(\Refob y A A),\
\Activated(\Refob y A A)\}$. The actor's knowledge set includes a
refob to itself and a refob to an external actor $E$. $A$ can
become a receptionist by sending $E$ a refob to itself.
Henceforth, we will only consider configurations that are derivable
from an initial configuration.

\subsection{Standard Actor Operations}\label{sec:standard-actor-operations}

\begin{figure}[t]
$\textsc{Spawn}(x, A, B)$ 
$$\Config{\alpha, [\Phi]_A}{\mu}{\rho}{\chi} \InternalStep \Config{\alpha, [\Phi \cup \{ \Activated(\Refob x A B) \}]_A, [\Psi]_B}{\mu}{\rho}{\chi}$$
\begin{tabular}{ll}
where & $x,y,B$  fresh\\
and & $\Psi = \{ \Created(\Refob x A B),\ \Created(\Refob {y} B B),\ \Activated(\Refob y B B) \}$
\end{tabular}

\vspace{0.5cm}

$\textsc{Send}(x,\vec y, \vec z, A, B,\vec C)$ 
$$\Config{\alpha, [\Phi]_A}{\mu}{\rho}{\chi} \InternalStep \Config{\alpha, [\IncSent(x,\Phi) \cup \Psi]_A}{\mu, \Msg{B}{\AppMsg(x,R)}}{\rho}{\chi}$$
\begin{tabular}{ll}
where & $\vec y$ and $\vec z$ fresh and $n = |\vec y| = |\vec z| = |\vec C|$\\
and & $\Phi \vdash \Activated(\Refob x A B)$ and $\forall i \le n,\ \Phi \vdash \Activated(\Refob{y_i}{A}{C_i})$\\
and & $R = \{\Refob{z_i}{B}{C_i}\ |\ i \le n \}$ and $\Psi = \{\CreatedUsing(y_i,z_i)\ |\ i \le n \}$
\end{tabular}

\vspace{0.5cm}

$\textsc{Receive}(x,B,R)$ 
$$\Config{\alpha, (\Phi)_B}{\mu, \Msg{B}{\AppMsg(x,R)}}{\rho}{\chi} \InternalStep \Config{\alpha, [\IncRecv(x,\Phi) \cup \Psi]_B}{\mu}{\rho}{\chi}$$
\begin{tabular}{ll}
where $\Psi = \{\Activated(z)\ |\ z \in R\}$
\end{tabular}

\vspace{0.5cm}

$\textsc{Idle}(A)$ 
$$\Config{\alpha, [\Phi]_A}{\mu}{\rho}{\chi} \InternalStep \Config{\alpha, (\Phi)_A}{\mu}{\rho}{\chi}$$

  \caption{Rules for standard actor interactions.}
  \label{rules:actors}
\end{figure}

Fig.~\ref{rules:actors} gives transition rules for standard actor operations, such as spawning actors and sending messages. Each of these rules corresponds a rule in the standard operational semantics of actors~\cite{aghaFoundationActorComputation1997}. Note that each rule is atomic, but can just as well be implemented as a sequence of several smaller steps without loss of generality because actors do not share state -- see \cite{aghaFoundationActorComputation1997} for a formal proof.

The \textsc{Spawn} event allows a busy actor $A$ to spawn a new actor $B$ and creates two refobs $\Refob x A B,\ \Refob y B B$. $B$ is initialized with knowledge about $x$ and $y$ via the facts $\Created(x),\Created(y)$. The facts $\Activated(x), \Activated(y)$ allow $A$ and $B$ to immediately begin sending messages to $B$. Note that implementing \textsc{Spawn} does not require a synchronization protocol between $A$ and $B$ to construct $\Refob x A B$. The parent $A$ can pass both its address and the freshly generated token $x$ to the constructor for $B$. Since actors typically know their own addresses, this allows $B$ to construct the triple $(x,A,B)$. Since the \texttt{spawn} call typically returns the address of the spawned actor, $A$ can also create the same triple.

The \textsc{Send} event allows a busy actor $A$ to send an application-level message to $B$ containing a set of refobs $z_1,\dots,z_n$ to actors $\vec C = C_1,\dots,C_n$ -- it is possible that $B = A$ or $C_i = A$ for some $i$. For each new refob $z_i$, we say that the message \emph{contains $z_i$}. Any other data in the message besides these refobs is irrelevant to termination detection and therefore omitted. To send the message, $A$ must have active refobs to both the target actor $B$ and to every actor $C_1,\dots,C_n$ referenced in the message. For each target $C_i$, $A$ adds a fact $\CreatedUsing(y_i,z_i)$ to its knowledge set; we say that $A$ \emph{created $z_i$ using $y_i$}. Finally, $A$ must increment its $\SentCount$ count for the refob $x$ used to send the message; we say that the message is sent \emph{along $x$}.

The \textsc{Receive} event allows an idle actor $B$ to become busy by consuming an application message sent to $B$. Before  performing subsequent actions, $B$ increments the receive count for $x$ and adds all refobs in the message to its knowledge set.

Finally, the \textsc{Idle} event puts a busy actor into the idle state, enabling it to consume another message.

\subsection{Release Protocol}\label{sec:release-protocol}

\begin{figure}[t!]

$\textsc{SendInfo}(y,z,A,B,C)$ 
$$\Config{\alpha, [\Phi \cup \Psi]_A}{\mu}{\rho}{\chi} \InternalStep \Config{\alpha, [\IncSent(y,\Phi)]_A}{\mu,\Msg{C}{\InfoMsg(y,z,B)}}{\rho}{\chi}$$
\begin{tabular}{ll}
where $\Psi = \{\CreatedUsing(\Refob y A C,\Refob z B C)\}$
\end{tabular}

\vspace{0.5cm}

$\textsc{Info}(y,z,B,C)$ 
$$\Config{\alpha, (\Phi)_C}{\mu, \Msg{C}{\InfoMsg(y,z,B)}}{\rho}{\chi} \InternalStep \Config{\alpha, (\IncRecv(y,\Phi) \cup \Psi)_C}{\mu}{\rho}{\chi}$$
\begin{tabular}{ll}
where $\Psi = \{\Created(\Refob z B C)\}$
\end{tabular}

\vspace{0.5cm}

$\textsc{SendRelease}(x,A,B)$ 
$$\Config{\alpha, [\Phi \cup \Psi]_A}{\mu}{\rho}{\chi} \InternalStep \Config{\alpha, [\Phi]_A}{\mu, \Msg{B}{\ReleaseMsg(x,n)}}{\rho}{\chi}$$
\begin{tabular}{ll}
where &$\Psi = \{\Activated(\Refob x A B), \SentCount(x,n)\}$\\
and & $\not\exists y,\ \CreatedUsing(x,y) \in \Phi$
\end{tabular}

\vspace{0.5cm}

$\textsc{Release}(x,A,B)$
$$\Config{\alpha, (\Phi)_B}{\mu, \Msg{B}{\ReleaseMsg(x,n)}}{\rho}{\chi} \InternalStep \Config{\alpha, (\Phi \cup \{\Released(x)\})_B}{\mu}{\rho}{\chi}$$
\begin{tabular}{l}
only if $\Phi \vdash \RecvCount(x,n)$
\end{tabular}

\vspace{0.5cm}

$\textsc{Compaction}(x,B,C)$ 
$$\Config{\alpha, (\Phi \cup \Psi)_C}{\mu}{\rho}{\chi} \InternalStep \Config{\alpha, (\Phi)_C}{\mu}{\rho}{\chi}$$
\begin{tabular}{ll}
where & $\Psi = \{\Created(\Refob x B C), \Released(\Refob x B C), \RecvCount(x,n)\}$ for some $n \in \mathbb N$\\
or & $\Psi = \{\Created(\Refob x B C), \Released(\Refob x B C)\}$ and $\forall n \in \mathbb N,\ \RecvCount(x,n) \not\in \Phi$
\end{tabular}

\vspace{0.5cm}

$\textsc{Snapshot}(A, \Phi)$ 
$$\Config{\alpha, (\Phi)_A}{\mu}{\rho}{\chi} \InternalStep \Config{\alpha, (\Phi)_A}{\mu}{\rho}{\chi}$$

  \caption{Rules for performing the release protocol.}
  \label{rules:release}
\end{figure}

Whenever an actor creates or receives a refob, it adds facts to its knowledge set. To remove these facts when they are no longer needed, actors can perform the \emph{release protocol} defined in Fig.~\ref{rules:release}. All of these rules are not present in the standard operational semantics of actors.

The \textsc{SendInfo} event allows a busy actor $A$ to inform $C$ about a refob $\Refob z B C$ that it created using $y$; we say that the $\InfoMsg$ message is sent \emph{along $y$} and \emph{contains $z$}. This event allows $A$ to remove the fact $\CreatedUsing(y,z)$ from its knowledge set. It is crucial that $A$ also increments its $\SentCount$ count for $y$ to indicate an undelivered $\InfoMsg$ message sent to $C$: it allows the snapshot aggregator to detect when there are undelivered $\InfoMsg$ messages, which contain refobs. This message is delivered with the \textsc{Info} event, which adds the fact $\Created(\Refob z B C)$ to $C$'s knowledge set and correspondingly increments $C$'s $\RecvCount$ count for $y$.

When an actor $A$ no longer needs $\Refob x A B$ for sending messages, $A$ can deactivate $x$ with the \textsc{SendRelease} event; we say that the $\ReleaseMsg$ is sent \emph{along $x$}. A precondition of this event is that $A$ has already sent messages to inform $B$ about all the refobs it has created using $x$. In practice, an implementation may defer sending any $\InfoMsg$ or $\ReleaseMsg$ messages to a target $B$ until all $A$'s refobs to $B$ are deactivated. This introduces a trade-off between the number of control messages and the rate of simple garbage detection (Section~\ref{sec:chain-lemma}).

Each $\ReleaseMsg$ message for a refob $x$ includes a count $n$ of the number of messages sent using $x$. This ensures that $\ReleaseMsg(x,n)$ is only delivered after all the preceding messages sent along $x$ have been delivered. Once the \textsc{Release} event can be executed, it adds the fact that $x$ has been released to $B$'s knowledge set. Once $C$ has received both an $\InfoMsg$ and $\ReleaseMsg$ message for a refob $x$, it may remove facts about $x$ from its knowledge set using the \textsc{Compaction} event.

Finally, the \textsc{Snapshot} event captures an idle actor's knowledge set. For simplicity, we have omitted the process of disseminating snapshots to an aggregator. Although this event does not change the configuration, it allows us to prove properties about snapshot events at different points in time.

\subsection{Composition and Effects}\label{sec:actor-composition}

\begin{figure}
$\textsc{In}(A,R)$ 
$$\Config{\alpha}{\mu}{\rho}{\chi} \ExternalStep \Config{\alpha}{\mu, \Msg{A}{\AppMsg(\NullToken, R)}}{\rho}{\chi \cup \chi'}$$
\begin{tabular}{ll}
where & $A \in \rho$ and $R = \{ \Refob{x_1}{A}{B_1}, \dots, \Refob{x_n}{A}{B_n} \}$ and $x_1,\dots,x_n$ fresh\\
and & $\{B_1,\dots,B_n\} \cap \dom(\alpha) \subseteq \rho$ and $\chi' = \{B_1,\dots,B_n\} \setminus \dom(\alpha)$ \\
\end{tabular}

\vspace{0.5cm}

$\textsc{Out}(x,B,R)$
$$\Config{\alpha}{\mu,\Msg{B}{\AppMsg(x, R)}}{\rho}{\chi} \ExternalStep \Config{\alpha}{\mu}{\rho \cup \rho'}{\chi}$$
\begin{tabular}{ll}
where $B \in \chi$ and $R = \{ \Refob{x_1}{B}{C_1}, \dots, \Refob{x_n}{B}{C_n} \}$ and $\rho' = \{C_1,\dots,C_n\} \cap \dom(\alpha)$
\end{tabular}

\vspace{0.5cm}

$\textsc{ReleaseOut}(x,B)$ 
$$\Config{\alpha}{\mu,\Msg{B}{\ReleaseMsg(x,n)}}{\rho}{\chi \cup \{B\}} \ExternalStep \Config{\alpha}{\mu}{\rho}{\chi \cup \{B\}}$$

\vspace{0.2cm}

$\textsc{InfoOut}(y,z,A,B,C)$ 
$$\Config{\alpha}{\mu,\Msg{C}{\InfoMsg(y,z,A,B)}}{\rho}{\chi \cup \{C\}} \ExternalStep \Config{\alpha}{\mu}{\rho}{\chi \cup \{C\}}$$

  \caption{Rules for interacting with the outside world.}
  \label{rules:composition}
\end{figure}

We give rules to dictate how internal actors interact with external actors in
Fig.~\ref{rules:composition}. The \textsc{In} and \textsc{Out} rules correspond to similar rules in the standard operational semantics of actors.

Since internal garbage collection protocols are not exposed to the outside world, all $\ReleaseMsg$ and $\InfoMsg$ messages sent to external actors are simply dropped by the \textsc{ReleaseOut} and \textsc{InfoOut} events. Likewise, only $\AppMsg$ messages can enter the system. Since we cannot statically determine when a receptionist's address has been forgotten by all external actors, we assume that receptionists are never terminated. The resulting ``black box'' behavior of our system is the same as the actor systems in \cite{aghaFoundationActorComputation1997}. Hence, in principle DRL can be gradually integrated into a codebase by creating a subsystem for garbage-collected actors.

The \textsc{In} event allows an external actor to send an application-level message to a receptionist $A$ containing a set of refobs $R$, all owned by $A$. Since external actors do not use refobs, the message is sent using the special $\NullToken$ token. All targets in $R$ that are not internal actors are added to the set of external actors.

The \textsc{Out} event delivers an application-level message to an external actor with a set of refobs $R$. All internal actors referenced in $R$ become receptionists because their addresses have been exposed to the outside world.

\subsection{Garbage}\label{sec:garbage-defn}

We can now operationally characterize actor garbage in our model. An actor $A$ can \emph{potentially receive a message} in $\kappa$ if there is a sequence of events (possibly of length zero) leading from $\kappa$ to a configuration $\kappa'$ in which $A$ has an undelivered message. We say that an actor is \emph{terminated} if it is idle and cannot potentially receive a message.

An actor is \emph{blocked} if it satisfies three conditions: (1) it is idle, (2) it is not a receptionist, and (3) it has no undelivered messages; otherwise, it is \emph{unblocked}. We define \emph{potential reachability} as the reflexive transitive closure of the potential acquaintance relation. That is, $A_1$ can potentially reach $A_n$ if and only if there is a sequence of unreleased refobs $(\Refob {x_1} {A_1} {A_2}), \dots, (\Refob {x_n} {A_{n-1}} {A_n})$; recall that a refob $\Refob x A B$ is unreleased if its target $B$ has not yet received a $\ReleaseMsg$ message for $x$.

Notice that an actor can potentially receive a message if and only if it is potentially reachable from an unblocked actor. Hence an actor is terminated if and only if it is only potentially reachable by blocked actors. A special case of this is \emph{simple garbage}, in which an actor is blocked and has no potential inverse acquaintances besides itself.

We say that a set of actors $S$ is \emph{closed} (with respect to the potential inverse acquaintance relation) if, whenever $B \in S$ and there is an unreleased refob $\Refob x A B$, then also $A \in S$. Notice that the closure of a set of terminated actors is also a set of terminated actors.

\section{Chain Lemma}\label{sec:chain-lemma}

To determine if an actor has terminated, one must show that all of its potential inverse acquaintances have terminated. This appears to pose a problem for termination detection, since actors cannot have a complete listing of all their potential inverse acquaintances without some synchronization: actors would need to consult their acquaintances before creating new references to them. In this section, we show that the DRL protocol provides a weaker guarantee that will nevertheless prove sufficient: knowledge about an actor's refobs is \emph{distributed} across the system and there is always a ``path'' from the actor to any of its potential inverse acquaintances.

\begin{figure}
    \centering
    \begin{tikzpicture}
	\begin{pgfonlayer}{nodelayer}
		\node [style=actor] (0) at (-6.5, 2.75) {$A_2$};
		\node [style=actor] (1) at (-11.5, 1) {$A_1$};
		\node [style=none] (17) at (-6.5, -3.75) {$\mathtt{Created}(x_1)$};
		\node [style=none] (51) at (1.5, 1) {\textbf{(2)}};
		\node [style=none] (54) at (-13.5, 1) {\textbf{(1)}};
		\node [style=actor] (57) at (-1.5, 1) {$A_3$};
		\node [style=none] (58) at (-3.25, 2.45) {$\mathtt{app}(y,\{x_3\})$};
		\node [style=actor] (82) at (-6.5, -2.5) {$B$};
		\node [style=none] (83) at (-11.5, 2.25) {$\mathtt{CreatedUsing}(x_1,x_2)$};
		\node [style=none] (84) at (11, 0) {$\mathtt{info}(x_2,x_3)$};
		\node [style=actor] (85) at (8.5, 2.75) {$A_2$};
		\node [style=actor] (86) at (3.5, 1) {$A_1$};
		\node [style=none] (87) at (8.5, -3.75) {$\mathtt{Created}(x_1)$};
		\node [style=actor] (88) at (13.5, 1) {$A_3$};
		\node [style=none] (89) at (11.75, 2.45) {$\mathtt{app}(y,\{x_3\})$};
		\node [style=actor] (90) at (8.5, -2.5) {$B$};
		\node [style=none] (91) at (3.5, 2.25) {$\mathtt{CreatedUsing}(x_1,x_2)$};
		\node [style=none] (93) at (-6.5, 4) {$\mathtt{CreatedUsing}(x_2,x_3)$};
		\node [style=none] (94) at (-9.5, -1) {$x_1$};
		\node [style=none] (97) at (-7, 0) {$x_2$};
		\node [style=none] (98) at (8, 0) {$x_2$};
		\node [style=none] (99) at (5.5, -1) {$x_1$};
	\end{pgfonlayer}
	\begin{pgfonlayer}{edgelayer}
		\draw [style=reference] (1) to (82);
		\draw [style=message] (0) to (57);
		\draw [style=reference] (86) to (90);
		\draw [style=message] (85) to (88);
		\draw [style=message, bend left=15] (85) to (90);
		\draw [style=reference] (0) to (82);
		\draw [style=reference] (85) to (90);
	\end{pgfonlayer}
\end{tikzpicture}
    \caption{An example of a chain from $B$ to $x_3$.}
    \label{fig:chain-example}
\end{figure}
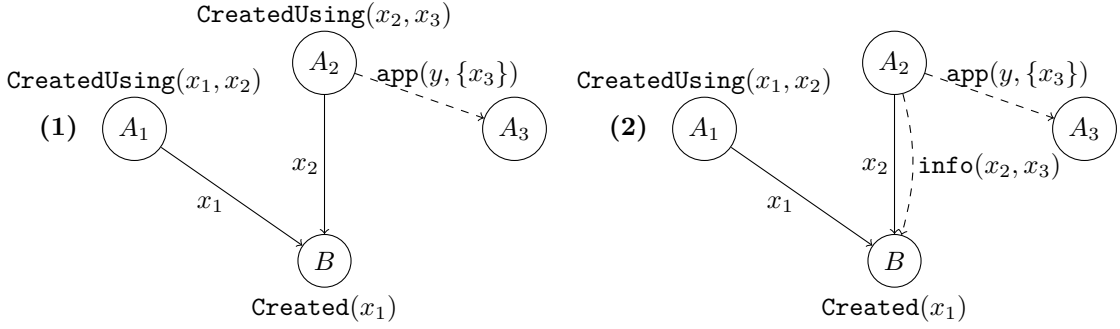

Let us construct a concrete example of such a path, depicted by Fig.~\ref{fig:chain-example}. Suppose that $A_1$ spawns $B$, gaining a refob $\Refob{x_1}{A_1}{B}$. Then $A_1$ may use $x_1$ to create $\Refob{x_2}{A_2}{B}$, which $A_2$ may receive and then use $x_2$ to create $\Refob{x_3}{A_3}{B}$. 

At this point, there are unreleased refobs owned by $A_2$ and $A_3$ that are not included in $B$'s knowledge set. However, Fig.~\ref{fig:chain-example} shows that the distributed knowledge of $B,A_1,A_2$ creates a ``path'' to all of $B$'s potential inverse acquaintances. Since $A_1$ spawned $B$, $B$ knows the fact $\Created(x_1)$. Then when $A_1$ created $x_2$, it added the fact $\CreatedUsing(x_1, x_2)$ to its knowledge set, and likewise $A_2$ added the fact $\CreatedUsing(x_2, x_3)$; each fact points to another actor that owns an unreleased refob to $B$ (Fig.~\ref{fig:chain-example} (1)). 

Since actors can remove $\CreatedUsing$ facts by sending $\InfoMsg$ messages, we also consider (Fig.~\ref{fig:chain-example} (2)) to be a ``path'' from $B$ to $A_3$. But notice that, once $B$ receives the $\InfoMsg$ message, the fact $\Created(x_3)$ will be added to its knowledge set and so there will be a ``direct path'' from $B$ to $A_3$. We formalize this intuition with the notion of a \emph{chain} in a given configuration $\Config{\alpha}{\mu}{\rho}{\chi}$:
\begin{definition}
A \emph{chain to $\Refob x A B$} is a sequence of unreleased refobs $(\Refob{x_1}{A_1}{B}),\allowbreak \dots,\allowbreak (\Refob{x_n}{A_n}{B})$ such that:
\begin{itemize}
    \item $\alpha(B) \vdash \Created(\Refob{x_1}{A_1}{B})$;
    \item For all $i < n$, either $\alpha(A_i) \vdash \CreatedUsing(x_i,x_{i+1})$ or the message $\Msg{B}{\InfoMsg(x_i,x_{i+1})}$ is in transit; and
    \item $A_n = A$ and $x_n = x$.
\end{itemize}
\end{definition}

We say that an actor $B$ is \emph{in the root set} if it is a receptionist or if there is an application message $\AppMsg(x,R)$ in transit to an external actor with $B \in \text{targets}(R)$. Since external actors never release refobs, actors in the root set must never terminate.

\begin{restatable}[Chain Lemma]{lemma}{ChainLemma}
\label{lem:chain-lemma}
Let $B$ be an internal actor in $\kappa$. If $B$ is not in the root set, then there is a chain to every unreleased refob $\Refob x A B$. Otherwise, there is a chain to some refob $\Refob y C B$ where $C$ is an external actor.
\end{restatable}
\begin{remark*}
When $B$ is in the root set, not all of its unreleased refobs are guaranteed to have chains. This is because an external actor may send $B$'s address to other receptionists without sending an $\InfoMsg$ message to $B$.
\end{remark*}

An immediate application of the Chain Lemma is to allow actors to detect when they are simple garbage. If any actor besides $B$ owns an unreleased refob to $B$, then $B$ must have a fact $\Created(\Refob x A B)$ in its knowledge set where $A \ne B$. Hence, if $B$ has no such facts, then it must have no nontrivial potential inverse acquaintances. Moreover, since actors can only have undelivered messages along unreleased refobs, $B$ also has no undelivered messages from any other actor; it can only have undelivered messages that it sent to itself. This gives us the following result:

\begin{theorem}
    Suppose $B$ is idle with knowledge set $\Phi$, such that:
    \begin{itemize}
        \item $\Phi$ does not contain any facts of the form $\Created(\Refob x A B)$ where $A \ne B$; and
        \item for all facts $\Created(\Refob x B B) \in \Phi$, also $\Phi \vdash \SentCount(x,n) \land \RecvCount(x,n)$ for some $n$.
    \end{itemize}
    Then $B$ is simple garbage.
\end{theorem}

\section{Termination Detection}\label{sec:termination-detection}

In order to detect non-simple terminated garbage, actors periodically sends a snapshot of their knowledge set to a snapshot aggregator actor.  An aggregator in turn may disseminate snapshots it has to other aggregators.  Each aggregator maintains a map data structure, associating an actor’s address to its most recent snapshot; in effect, snapshot aggregators maintain an eventually consistent key-value store with addresses as keys and snapshots as values. At any time, an aggregator can scan its local store to find terminated actors and send them a request to self-destruct.

Given an arbitrary set of snapshots $Q$, we characterize the \emph{finalized subsets} of $Q$ in this section. We show that the actors that took these finalized snapshots must be terminated. Conversely, the snapshots of any closed set of terminated actors are guaranteed to be finalized. (Recall that the closure of a set of terminated actors is also a terminated set of actors.) Thus, snapshot aggregators can eventually detect all terminated actors by periodically searching their local stores for finalized subsets. Finally, we give an algorithm for obtaining the maximum finalized subset of a set $Q$ by ``pruning away’’ the snapshots of actors that appear not to have terminated.

Recall that when we speak of a set of snapshots $Q$, we assume each snapshot was taken by a different actor. We will write $\Phi_A \in Q$ to denote $A$'s snapshot in $Q$; we will also write $A \in Q$ if $A$ has a snapshot in $Q$. We will also write $Q \vdash \phi$ if $\Phi \vdash \phi$ for some $\Phi \in Q$.

\begin{definition}
A set of snapshots $Q$ is \emph{closed} if, whenever $Q \vdash \Unreleased(\Refob x A B)$ and $B \in Q$, then also $A\in Q$ and $\Phi_A \vdash \Activated(\Refob x A B)$.
\end{definition}
\begin{definition}
An actor $B \in Q$ \emph{appears blocked} if, for every $Q \vdash \Unreleased(\Refob x A B)$, then $\Phi_A,\Phi_B \in Q$ and $\Phi_A \vdash \SentCount(x,n)$ and $\Phi_B \vdash \RecvCount(x,n)$ for some $n$.
\end{definition}
\begin{definition}
A set of snapshots $Q$ is \emph{finalized} if it is closed and every actor in $Q$ appears blocked. 
\end{definition}

This definition corresponds to our characterization in Section~\ref{sec:garbage-defn}: An actor is terminated precisely when it is in a closed set of blocked actors.

\begin{restatable}[Safety]{theorem}{Safety}\label{thm:safety}
If $Q$ is a finalized set of snapshots at time $t_f$ then the actors in $Q$ are all terminated at $t_f$.
\end{restatable}

We say that the \emph{final action} of a terminated actor is the last non-snapshot event it performs before becoming terminated. Notice that an actor's final action can only be an \textsc{Idle}, \textsc{Info}, or \textsc{Release} event. Note also that the final action may come \emph{strictly before} an actor becomes terminated, since a blocked actor may only terminate after all of its potential inverse acquaintances become blocked.

The following lemma allows us to prove that DRL is eventually live. It also shows that an non-finalized set of snapshots must have an unblocked actor.
\begin{restatable}{lemma}{Completeness}\label{lem:terminated-is-complete}
    Let $S$ be a closed set of terminated actors at time $t_f$. If every actor in $S$ took a snapshot sometime after its final action, then the resulting set of snapshots is finalized.
\end{restatable}

\begin{theorem}[Liveness]\label{thm:liveness}
If every actor eventually takes a snapshot after performing an \textsc{Idle}, \textsc{Info}, or \textsc{Release} event, then every terminated actor is eventually part of a finalized set of snapshots.
\end{theorem}
\begin{proof}
    If an actor $A$ is terminated, then the closure $S$ of $\{A\}$ is a terminated set of actors. Since every actor eventually takes a snapshot after taking its final action, \cref{lem:terminated-is-complete} implies that the resulting snapshots of $S$ are finalized.
\end{proof}

We say that a refob $\Refob x A B$ is \emph{unreleased} in $Q$ if $Q \vdash \Unreleased(x)$. Such a refob is said to be \emph{relevant} when $B \in Q$ implies $A \in Q$ and $\Phi_A \vdash \Activated(x)$ and $\Phi_A \vdash \SentCount(x,n)$ and $\Phi_B \vdash \RecvCount(x,n)$ for some $n$; intuitively, this indicates that $B$ has no undelivered messages along $x$. Notice that a set $Q$ is finalized if and only if all unreleased refobs in $Q$ are relevant.

Observe that if $\Refob x A B$ is unreleased and irrelevant in $Q$, then $B$ cannot be in any finalized subset of $Q$. We can therefore employ a simple iterative algorithm to find the maximum finalized subset of $Q$: for each irrelevant unreleased refob $\Refob x A B$ in $Q$, remove the target $B$ from $Q$. Since this can make another unreleased refob $\Refob y B C$ irrelevant, we must repeat this process until a fixed point is reached. In the resulting subset $Q'$, all unreleased refobs are relevant. Since all actors in $Q \setminus Q'$ are not members of any finalized subset of $Q$, it must be that $Q'$ is the maximum finalized subset of $Q$.

\section{Conclusion and Future Work}\label{sec:conclusion}\label{sec:future-work}

We have shown how deferred reference listing and message counts can be used to detect termination in actor systems. The technique is provably safe (Theorem~\ref{thm:safety}) and eventually live (Theorem~\ref{thm:liveness}). An implementation in Akka is presently underway.

We believe that DRL satisfies our three initial goals:
\begin{enumerate}
    \item \emph{Termination detection does not restrict concurrency in the application.} Actors do not need to coordinate their snapshots or pause execution during garbage collection.
    \item \emph{Termination detection does not impose high overhead.} The amortized memory overhead of our technique is linear in the number of unreleased refobs. Besides application messages, the only additional control messages required by the DRL communication protocol are $\InfoMsg$ and $\ReleaseMsg$ messages. These control messages can be batched together and deferred, at the cost of worse termination detection time.
    \item \emph{Termination detection scales with the number of nodes in the system.} Our algorithm is incremental, decentralized, and does not require synchronization between nodes.
\end{enumerate}

Since it does not matter what order snapshots are collected in, DRL can be used as a ``building block’’ for more sophisticated garbage collection algorithms. One promising direction is to take a \emph{generational} approach \cite{DBLP:journals/cacm/LiebermanH83}, in which long-lived actors take snapshots less frequently than short-lived actors. Different types of actors could also take snapshots at different rates. In another approach, snapshot aggregators could \emph{request} snapshots instead of waiting to receive them.

In the presence of faults, DRL remains safe but its liveness properties are affected. If an actor $A$ crashes and its state cannot be recovered, then none of its refobs can be released and the aggregator will never receive its snapshot. Consequently, all actors potentially reachable from $A$ can no longer be garbage collected. However, $A$'s failure does not affect the garbage collection of actors it cannot reach. In particular, network partitions between nodes will not delay node-local garbage collection.

Choosing an adequate fault-recovery protocol will likely vary depending on the target actor framework. One option is to use checkpointing or event-sourcing to persist GC state; the resulting overhead may be acceptable in applications that do not frequently spawn actors or create refobs. Another option is to monitor actors for failure and infer which refobs are no longer active; this is a subject for future work.

Another issue that can affect liveness is message loss: If any messages along a refob $\Refob x A B$ are dropped, then $B$ can never be garbage collected because it will always appear unblocked. This is, in fact, the desired behavior if one cannot guarantee that the message will not be delivered at some later point. In practice, this problem might be addressed with watermarking.

\bibliography{contents/bibliography}

\begin{thebibliography}{10}

\bibitem{books/daglib/0066897}
Gul Agha.
\newblock {\em {{ACTORS}} - a Model of Concurrent Computation in Distributed
  Systems}.
\newblock {{MIT Press}} Series in Artificial Intelligence. {MIT Press}, 1990.

\bibitem{journals/cacm/Agha90}
Gul Agha.
\newblock Concurrent object-oriented programming.
\newblock {\em Communications of the ACM}, 33(9):125--141, September 1990.

\bibitem{aghaFoundationActorComputation1997}
Gul~A. Agha, Ian~A. Mason, Scott~F. Smith, and Carolyn~L. Talcott.
\newblock A foundation for actor computation.
\newblock {\em Journal of Functional Programming}, 7(1):1--72, January 1997.
\newblock \href {https://doi.org/10.1017/S095679689700261X}
  {\path{doi:10.1017/S095679689700261X}}.

\bibitem{Akka}
Akka.
\newblock https://akka.io/.

\bibitem{armstrongConcurrentProgrammingERLANG1996}
Joe Armstrong, Robert Virding, Claes Wikstr{\"o}m, and Mike Williams.
\newblock {\em Concurrent Programming in {{ERLANG}}}.
\newblock {Prentice Hall}, {Englewood Cliffs, New Jersey}, second edition,
  1996.

\bibitem{bevanDistributedGarbageCollection1987}
Di~Bevan.
\newblock Distributed garbage collection using reference counting.
\newblock In G.~Goos, J.~Hartmanis, D.~Barstow, W.~Brauer, P.~Brinch~Hansen,
  D.~Gries, D.~Luckham, C.~Moler, A.~Pnueli, G.~Seegm{\"u}ller, J.~Stoer,
  N.~Wirth, J.~W. Bakker, A.~J. Nijman, and P.~C. Treleaven, editors, {\em
  {{PARLE Parallel Architectures}} and {{Languages Europe}}}, volume 259, pages
  176--187. {Springer Berlin Heidelberg}, {Berlin, Heidelberg}, 1987.
\newblock \href {https://doi.org/10.1007/3-540-17945-3_10}
  {\path{doi:10.1007/3-540-17945-3_10}}.

\bibitem{blessingTreeTopologiesCausal2017}
Sebastian Blessing, Sylvan Clebsch, and Sophia Drossopoulou.
\newblock Tree topologies for causal message delivery.
\newblock In {\em Proceedings of the 7th {{ACM SIGPLAN International Workshop}}
  on {{Programming Based}} on {{Actors}}, {{Agents}}, and {{Decentralized
  Control}} - {{AGERE}} 2017}, pages 1--10, {Vancouver, BC, Canada}, 2017. {ACM
  Press}.
\newblock \href {https://doi.org/10.1145/3141834.3141835}
  {\path{doi:10.1145/3141834.3141835}}.

\bibitem{bykovOrleansCloudComputing2011}
Sergey Bykov, Alan Geller, Gabriel Kliot, James~R. Larus, Ravi Pandya, and
  Jorgen Thelin.
\newblock Orleans: Cloud computing for everyone.
\newblock In {\em Proceedings of the 2nd {{ACM Symposium}} on {{Cloud
  Computing}} - {{SOCC}} '11}, pages 1--14, {Cascais, Portugal}, 2011. {ACM
  Press}.
\newblock \href {https://doi.org/10.1145/2038916.2038932}
  {\path{doi:10.1145/2038916.2038932}}.

\bibitem{chandyDistributedSnapshotsDetermining1985}
K.~Mani Chandy and Leslie Lamport.
\newblock Distributed snapshots: Determining global states of distributed
  systems.
\newblock {\em ACM Transactions on Computer Systems}, 3(1):63--75, February
  1985.
\newblock \href {https://doi.org/10.1145/214451.214456}
  {\path{doi:10.1145/214451.214456}}.

\bibitem{clebschFullyConcurrentGarbage2013}
Sylvan Clebsch and Sophia Drossopoulou.
\newblock Fully concurrent garbage collection of actors on many-core machines.
\newblock In {\em Proceedings of the 2013 {{ACM SIGPLAN}} International
  Conference on {{Object}} Oriented Programming Systems Languages \&
  Applications - {{OOPSLA}} '13}, pages 553--570, {Indianapolis, Indiana, USA},
  2013. {ACM Press}.
\newblock \href {https://doi.org/10.1145/2509136.2509557}
  {\path{doi:10.1145/2509136.2509557}}.

\bibitem{fidge1987timestamps}
Colin~J Fidge.
\newblock Timestamps in message-passing systems that preserve the partial
  ordering.
\newblock {\em Australian Computer Science Communications}, 10(1):56--66,
  February 1988.

\bibitem{Baker-Hewitt-laws77}
Carl Hewitt and Henry~G. Baker.
\newblock Laws for communicating parallel processes.
\newblock In Bruce Gilchrist, editor, {\em Information Processing, Proceedings
  of the 7th {{IFIP}} Congress 1977, Toronto, Canada, August 8-12, 1977}, pages
  987--992. {North-Holland}, 1977.

\bibitem{kafuraConcurrentDistributedGarbage1995}
D.~Kafura, M.~Mukherji, and D.M. Washabaugh.
\newblock Concurrent and distributed garbage collection of active objects.
\newblock {\em IEEE Transactions on Parallel and Distributed Systems},
  6(4):337--350, April 1995.
\newblock \href {https://doi.org/10.1109/71.372788}
  {\path{doi:10.1109/71.372788}}.

\bibitem{laiTerminationDetectionDynamically1986}
Ten-Hwang Lai.
\newblock Termination detection for dynamically distributed systems with
  non-first-in-first-out communication.
\newblock {\em Journal of Parallel and Distributed Computing}, 3(4):577--599,
  December 1986.
\newblock \href {https://doi.org/10.1016/0743-7315(86)90015-8}
  {\path{doi:10.1016/0743-7315(86)90015-8}}.

\bibitem{DBLP:journals/cacm/LiebermanH83}
Henry Lieberman and Carl Hewitt.
\newblock A real-time garbage collector based on the lifetimes of objects.
\newblock {\em Commun. {ACM}}, 26(6):419--429, 1983.
\newblock \href {https://doi.org/10.1145/358141.358147}
  {\path{doi:10.1145/358141.358147}}.

\bibitem{matochaTaxonomyDistributedTermination1998}
Jeff Matocha and Tracy Camp.
\newblock A taxonomy of distributed termination detection algorithms.
\newblock {\em Journal of Systems and Software}, 43(3):207--221, November 1998.
\newblock \href {https://doi.org/10.1016/S0164-1212(98)10034-1}
  {\path{doi:10.1016/S0164-1212(98)10034-1}}.

\bibitem{matternAlgorithmsDistributedTermination1987}
Friedemann Mattern.
\newblock Algorithms for distributed termination detection.
\newblock {\em Distributed Computing}, 2(3):161--175, September 1987.
\newblock \href {https://doi.org/10.1007/BF01782776}
  {\path{doi:10.1007/BF01782776}}.

\bibitem{NHSDeployRiak2013}
{{NHS}} to {{Deploy Riak}} for {{New IT Backbone With Quality}} of {{Care
  Improvements}} in {{Sight}}.
\newblock
  https://riak.com/nhs-to-deploy-riak-for-new-it-backbone-with-quality-of-care-improvements-in-sight.html,
  October 2013.

\bibitem{PayPalBlowsBillion}
{{PayPal Blows Past}} 1 {{Billion Transactions Per Day Using Just}} 8 {{VMs
  With Akka}}, {{Scala}}, {{Kafka}} and {{Akka Streams}}.
\newblock
  https://www.lightbend.com/case-studies/paypal-blows-past-1-billion-transactions-per-day-using-just-8-vms-and-akka-scala-kafka-and-akka-streams.

\bibitem{piquerIndirectReferenceCounting1991}
Jos{\'e}~M. Piquer.
\newblock Indirect {{Reference Counting}}: {{A Distributed Garbage Collection
  Algorithm}}.
\newblock In Emile H.~L. Aarts, Jan {van Leeuwen}, and Martin Rem, editors,
  {\em Parle '91 {{Parallel Architectures}} and {{Languages Europe}}}, volume
  505, pages 150--165. {Springer Berlin Heidelberg}, {Berlin, Heidelberg},
  1991.
\newblock \href {https://doi.org/10.1007/978-3-662-25209-3_11}
  {\path{doi:10.1007/978-3-662-25209-3_11}}.

\bibitem{DBLP:conf/iwmm/PlainfosseS95}
David Plainfoss{\'e} and Marc Shapiro.
\newblock A survey of distributed garbage collection techniques.
\newblock In {\em Memory Management, International Workshop {{IWMM}} 95,
  Kinross, {{UK}}, September 27-29, 1995, Proceedings}, pages 211--249, 1995.
\newblock \href {https://doi.org/10.1007/3-540-60368-9\\_26}
  {\path{doi:10.1007/3-540-60368-9\\_26}}.

\bibitem{plyukhinConcurrentGarbageCollection2018}
Dan Plyukhin and Gul Agha.
\newblock Concurrent garbage collection in the actor model.
\newblock In {\em Proceedings of the 8th {{ACM SIGPLAN International Workshop}}
  on {{Programming Based}} on {{Actors}}, {{Agents}}, and {{Decentralized
  Control}} - {{AGERE}} 2018}, pages 44--53, {Boston, MA, USA}, 2018. {ACM
  Press}.
\newblock \href {https://doi.org/10.1145/3281366.3281368}
  {\path{doi:10.1145/3281366.3281368}}.

\bibitem{schelvisIncrementalDistributionTimestamp1989}
M.~Schelvis.
\newblock Incremental distribution of timestamp packets: A new approach to
  distributed garbage collection.
\newblock In {\em Conference Proceedings on {{Object}}-Oriented Programming
  Systems, Languages and Applications - {{OOPSLA}} '89}, pages 37--48, {New
  Orleans, Louisiana, United States}, 1989. {ACM Press}.
\newblock \href {https://doi.org/10.1145/74877.74883}
  {\path{doi:10.1145/74877.74883}}.

\bibitem{vardhanUsingPassiveObject2003}
Abhay Vardhan and Gul Agha.
\newblock Using passive object garbage collection algorithms for garbage
  collection of active objects.
\newblock {\em ACM SIGPLAN Notices}, 38(2 supplement):106, February 2003.
\newblock \href {https://doi.org/10.1145/773039.512443}
  {\path{doi:10.1145/773039.512443}}.

\bibitem{venkatasubramanianScalableDistributedGarbage1992}
Nalini Venkatasubramanian, Gul Agha, and Carolyn Talcott.
\newblock Scalable distributed garbage collection for systems of active
  objects.
\newblock In Yves Bekkers and Jacques Cohen, editors, {\em Memory
  {{Management}}}, volume 637, pages 134--147. {Springer-Verlag},
  {Berlin/Heidelberg}, 1992.
\newblock \href {https://doi.org/10.1007/BFb0017187}
  {\path{doi:10.1007/BFb0017187}}.

\bibitem{venkatasubramanianReasoningMetaLevel1995}
Nalini Venkatasubramanian and Carolyn Talcott.
\newblock Reasoning about meta level activities in open distributed systems.
\newblock In {\em Proceedings of the Fourteenth Annual {{ACM}} Symposium on
  {{Principles}} of Distributed Computing - {{PODC}} '95}, pages 144--152,
  {Ottowa, Ontario, Canada}, 1995. {ACM Press}.
\newblock \href {https://doi.org/10.1145/224964.224981}
  {\path{doi:10.1145/224964.224981}}.

\bibitem{vishnevskiyHowDiscordScaled2017}
Stanislav Vishnevskiy.
\newblock How {{Discord Scaled Elixir}} to 5,000,000 {{Concurrent Users}}.
\newblock https://blog.discord.com/scaling-elixir-f9b8e1e7c29b, July 2017.

\bibitem{wangConservativeSnapshotbasedActor2011}
Wei-Jen Wang.
\newblock Conservative snapshot-based actor garbage collection for distributed
  mobile actor systems.
\newblock {\em Telecommunication Systems}, June 2011.
\newblock \href {https://doi.org/10.1007/s11235-011-9509-1}
  {\path{doi:10.1007/s11235-011-9509-1}}.

\bibitem{wangActorGarbageCollection2010}
Wei-Jen Wang, Carlos Varela, Fu-Hau Hsu, and Cheng-Hsien Tang.
\newblock Actor {{Garbage Collection Using Vertex}}-{{Preserving
  Actor}}-to-{{Object Graph Transformations}}.
\newblock In David Hutchison, Takeo Kanade, Josef Kittler, Jon~M. Kleinberg,
  Friedemann Mattern, John~C. Mitchell, Moni Naor, Oscar Nierstrasz,
  C.~Pandu~Rangan, Bernhard Steffen, Madhu Sudan, Demetri Terzopoulos, Doug
  Tygar, Moshe~Y. Vardi, Gerhard Weikum, Paolo Bellavista, Ruay-Shiung Chang,
  Han-Chieh Chao, Shin-Feng Lin, and Peter M.~A. Sloot, editors, {\em Advances
  in {{Grid}} and {{Pervasive Computing}}}, volume 6104, pages 244--255.
  {Springer Berlin Heidelberg}, {Berlin, Heidelberg}, 2010.
\newblock \href {https://doi.org/10.1007/978-3-642-13067-0_28}
  {\path{doi:10.1007/978-3-642-13067-0_28}}.

\bibitem{wangDistributedGarbageCollection2006}
Wei-Jen Wang and Carlos~A. Varela.
\newblock Distributed {{Garbage Collection}} for {{Mobile Actor Systems}}:
  {{The Pseudo Root Approach}}.
\newblock In Yeh-Ching Chung and Jos{\'e}~E. Moreira, editors, {\em Advances in
  {{Grid}} and {{Pervasive Computing}}}, volume 3947, pages 360--372. {Springer
  Berlin Heidelberg}, {Berlin, Heidelberg}, 2006.
\newblock \href {https://doi.org/10.1007/11745693_36}
  {\path{doi:10.1007/11745693_36}}.

\bibitem{watsonEfficientGarbageCollection1987}
Paul Watson and Ian Watson.
\newblock An efficient garbage collection scheme for parallel computer
  architectures.
\newblock In G.~Goos, J.~Hartmanis, D.~Barstow, W.~Brauer, P.~Brinch~Hansen,
  D.~Gries, D.~Luckham, C.~Moler, A.~Pnueli, G.~Seegm{\"u}ller, J.~Stoer,
  N.~Wirth, J.~W. Bakker, A.~J. Nijman, and P.~C. Treleaven, editors, {\em
  {{PARLE Parallel Architectures}} and {{Languages Europe}}}, volume 259, pages
  432--443. {Springer Berlin Heidelberg}, {Berlin, Heidelberg}, 1987.
\newblock \href {https://doi.org/10.1007/3-540-17945-3_25}
  {\path{doi:10.1007/3-540-17945-3_25}}.

\end{thebibliography}

\appendix

\section{Appendix}

\subsection{Basic Properties}

\begin{lemma}\label{lem:release-is-final}
If $B$ has undelivered messages along $\Refob x A B$, then $x$ is an unreleased refob.
\end{lemma}
\begin{proof}
    There are three types of messages: $\AppMsg, \InfoMsg,$ and $\ReleaseMsg$. All three messages can only be sent when $x$ is active. Moreover, the \textsc{Release} rule ensures that they must all be delivered before $x$ can be released.
\end{proof}

\begin{lemma}\label{lem:facts-remain-until-cancelled}
$\ $
\begin{itemize}
    \item Once $\CreatedUsing(\Refob y A C, \Refob z B C)$ is added to $A$'s knowledge set, it will not be removed until after $A$ has sent an $\InfoMsg$ message containing $z$ to $C$.

    \item Once $\Created(\Refob z B C)$ is added to $C$'s knowledge set, it will not be removed until after $C$ has received the (unique) $\ReleaseMsg$ message along $z$.

    \item Once $\Released(\Refob z B C)$ is added to $C$'s knowledge set, it will not be removed until after $C$ has received the (unique) $\InfoMsg$ message containing $z$.
\end{itemize}
\end{lemma}
\begin{proof}
    Immediate from the transition rules.
\end{proof}

\begin{lemma}\label{lem:msg-counts}
    Consider a refob $\Refob x A B$.  Let $t_1, t_2$ be times such that $x$ has not yet been deactivated at $t_1$ and $x$ has not yet been released at $t_2$. In particular, $t_1$ and $t_2$ may be before the creation time of $x$.
    
    Suppose that $\alpha_{t_1}(A) \vdash \SentCount(x,n)$ and $\alpha_{t_2}(B) \vdash \RecvCount(x,m)$ and, if $t_1 < t_2$, that $A$ does not send any messages along $x$ during the interval $[t_1,t_2]$ . Then the difference $\max(n - m,0)$ is the number of messages sent along $x$ before $t_1$ that were not received before $t_2$.
\end{lemma}
\begin{proof}
    Since $x$ is not deactivated at time $t_1$ and unreleased at time $t_2$, the message counts were never reset by the \textsc{SendRelease} or \textsc{Compaction} rules. Hence $n$ is the number of messages $A$ sent along $x$ before $t_1$ and $m$ is the number of messages $B$ received along $x$ before $t_2$. Hence $\max(n - m, 0)$ is the number of messages sent before $t_1$ and \emph{not} received before $t_2$.
\end{proof}

\subsection{Chain Lemma}

\ChainLemma*
\begin{proof}
    We prove that the invariant holds in the initial configuration and at all subsequent times by induction on events $\kappa \Step e \kappa'$, omitting events that do not affect chains. Let $\kappa = \Config{\alpha}{\mu}{\rho}{\chi}$ and $\kappa' = \Config{\alpha'}{\mu'}{\rho'}{\chi'}$.
    
    In the initial configuration, the only refob to an internal actor is $\Refob y A A$. Since $A$ knows $\Created(\Refob{y}{A}{A})$, the invariant is satisfied.
    
    In the cases below, let $x,y,z,A,B,C$ be free variables, not referencing the variables used in the statement of the lemma.
        
    \begin{itemize}
        \item $\textsc{Spawn}(x,A,B)$ creates a new unreleased refob $\Refob x A B$, which satisfies the invariant because $\alpha'(B) \vdash \Created(\Refob x A B)$.

        \item $\textsc{Send}(x,\vec y, \vec z, A,B,\vec C)$ creates a set of refobs $R$. Let $(\Refob z B C) \in R$, created using $\Refob y A C$.
        
        If $C$ is already in the root set, then the invariant is trivially preserved. Otherwise, there must be a chain $(\Refob{x_1}{A_1}{C}), \dots, (\Refob{x_n}{A_n}{C})$ where $x_n = y$ and $A_n = A$. Then $x_1,\dots,x_n,z$ is a chain in $\kappa'$, since $\alpha'(A_n) \vdash \CreatedUsing(x_n,z)$. 
        
        If $B$ is an internal actor, then this shows that every unreleased refob to $C$ has a chain in $\kappa'$. Otherwise, $C$ is in the root set in $\kappa'$. To see that the invariant still holds, notice that $\Refob z B C$ is a witness of the desired chain.
        
        \item $\textsc{SendInfo}(y,z,A,B,C)$ removes the $\CreatedUsing(y,z)$ fact but also sends $\InfoMsg(y,z,B)$, so chains are unaffected.
        
        \item $\textsc{Info}(y,z,B,C)$ delivers $\InfoMsg(y,z,B)$ to $C$ and adds $\Created(\Refob z B C)$ to its knowledge set.
        
        Suppose $\Refob z B C$ is part of a chain $(\Refob{x_1}{A_1}{C}), \dots, (\Refob{x_n}{A_n}{C})$, i.e. $x_i = y$ and $x_{i+1} = z$ and $A_{i+1} = B$ for some $i < n$. Since $\alpha'(C) \vdash \Created(\Refob{x_{i+1}}{A_{i+1}}{C})$, we still have a chain $x_{i+1},\dots,x_n$ in $\kappa'$.
        
        \item $\textsc{Release}(x,A,B)$ releases the refob $\Refob x A B$. Since external actors never release their refobs, both $A$ and $B$ must be internal actors.
        
        Suppose the released refob was part of a chain $(\Refob{x_1}{A_1}{B}), \dots, (\Refob{x_n}{A_n}{B})$, i.e. $x_i = x$ and $A_i = A$ for some $i < n$. We will show that $x_{i+1},\dots,x_n$ is a chain in $\kappa'$.
        
        Before performing $\textsc{SendRelease}(x_i,A_i,B)$, $A_i$ must have performed the $\textsc{Info}(x_i,x_{i+1},\allowbreak A_{i+1},B)$ event. Since the $\InfoMsg$ message was sent along $x_i$, Lemma~\ref{lem:release-is-final} ensures that the message must have been delivered before the present \textsc{Release} event. Furthermore, since $x_{i+1}$ is an unreleased refob in $\kappa'$, Lemma~\ref{lem:facts-remain-until-cancelled} ensures that $\alpha'(B) \vdash \Created(\Refob{x_{i+1}}{A_{i+1}}{B})$.
        
        \item $\textsc{In}(A,R)$ adds a message from an external actor to the internal actor $A$. This event can only create new refobs that point to receptionists, so it preserves the invariant.

        \item $\textsc{Out}(x,B,R)$ emits a message $\AppMsg(x,R)$ to the external actor $B$. Since all targets in $R$ are already in the root set, the invariant is preserved.
    \end{itemize}
\end{proof}

\subsection{Termination Detection}

Given a set of snapshots $Q$ taken before some time $t_f$, we write $Q_t$ to denote those snapshots in $Q$ that were taken before time $t < t_f$. If $\Phi_A \in Q$, we denote the time of $A$'s snapshot as $t_A$.

\Completeness*
Call this set of snapshots $Q$. First, we prove the following lemma.
\begin{lemma}\label{lem:completeness-helper}
    If $Q \vdash \Unreleased(\Refob x A B)$ and $B \in Q$, then $x$ is unreleased at $t_B$. 
\end{lemma}
\begin{proof}
    By definition, $Q \vdash \Unreleased(\Refob x A B)$ only if $Q \vdash \Created(x) \land \lnot \Released(x)$. Since $Q \not\vdash \Released(x)$, we must also have $\Phi_B \not\vdash \Released(x)$. For $Q \vdash \Created(x)$, there are two cases. 
        
    Case 1: $\Phi_B \vdash \Created(x)$. Since $\Phi_B \not\vdash \Released(x)$, \cref{lem:facts-remain-until-cancelled} implies that $x$ is unreleased at time $t_B$.
    
    Case 2: For some $C \in Q$ and some $y$, $\Phi_C \vdash \CreatedUsing(y,x)$. Since $C$ performed its final action before taking its snapshot, this implies that $C$ will never send the $\InfoMsg$ message containing $x$ to $B$. 
    
    Suppose then for a contradiction that $x$ is released at time $t_B$. Since $\Phi_B \not\vdash \Released(x)$, \cref{lem:facts-remain-until-cancelled} implies that $B$ received an $\InfoMsg$ message containing $x$ before its snapshot. But this is impossible because $C$ never sends this message.
\end{proof}

\begin{proof}[Proof (\cref{lem:terminated-is-complete})]
    By strong induction on time $t$, we show that $Q$ is closed and that every actor appears blocked.
    
    \textbf{Induction hypothesis:} For all times $t' < t$, if $B \in Q_{t'}$ and $Q \vdash \Unreleased(\Refob x A B)$, then $A \in Q$, $Q \vdash \Activated(x)$, and $Q \vdash \SentCount(x,n)$ and $Q \vdash \RecvCount(x,n)$ for some $n$.
    
    Since $Q_0 = \emptyset$, the induction hypothesis holds trivially in the initial configuration.
    
    Now assume the induction hypothesis. Suppose that $B \in Q$ takes its snapshot at time $t$ with $Q \vdash \Unreleased(\Refob x A B)$, which implies $Q \vdash \Created(x) \land \lnot\Released(x)$.
    
    $Q \vdash \Created(x)$ implies that $x$ was created before $t_f$. \cref{lem:completeness-helper} implies that $x$ is also unreleased at time $t_f$, since $B$ cannot perform a \textsc{Release} event after its final action. Hence $A$ is in the closure of $\{B\}$ at time $t_f$, so $A \in Q$.
    
    Now suppose $\Phi_A \not\vdash \Activated(x)$. Then either $x$ will be activated after $t_A$ or $x$ was deactivated before $t_A$. The former is impossible because $A$ would need to become unblocked to receive $x$. Since $x$ is unreleased at time $t_f$ and $t_A < t_f$, the latter implies that there is an undelivered $\ReleaseMsg$ message for $x$ at time $t_f$. But this is impossible as well, since $B$ is blocked at $t_f$.
    
    Finally, let $n$ such that $\Phi_B \vdash \RecvCount(x,n)$; we must show that $\Phi_A \vdash \SentCount(x,n)$. By the above arguments, $x$ is active at time $t_A$ and unreleased at time $t_B$. Since both actors performed their final action before their snapshots, all messages sent before $t_A$ must have been delivered before $t_B$. By Lemma~\ref{lem:msg-counts}, this implies $\Phi_A \vdash \SentCount(x,n)$.
\end{proof}

We now prove the safety theorem, which states that if $Q$ is a finalized set of snapshots, then the corresponding actors of $Q$ are terminated. We do this by showing that at each time $t$, all actors in $Q_t$ are blocked and all of their potential inverse acquaintances are in $Q$. %Interestingly, the actors in $Q \setminus Q_t$ might not be blocked at time $t$.

Consider the first actor $B$ in $Q$ to take a snapshot. We show, using the Chain Lemma, that the closure of this actor is in $Q$. Then, since all potential inverse acquaintances of $B$ take snapshots strictly after $t_B$, it is impossible for $B$ to have any undelivered messages without appearing unblocked.

For every subsequent actor $B$ to take a snapshot, we make a similar argument with an additional step: If $B$ has any potential inverse acquaintances in $Q_{t_B}$, then they could not have sent $B$ a message without first becoming unblocked.

\Safety*
\begin{proof}
Proof by induction on events. The induction hypothesis consists of two clauses that must both be satisfied at all times $t \le t_f$.
\begin{itemize}
    \item \textbf{IH 1:} If $B \in Q_t$ and $\Refob x A B$ is unreleased, then $Q \vdash \Unreleased(x)$.
    \item \textbf{IH 2:} The actors of $Q_t$ are all blocked.
\end{itemize}
%In other words, once an actor in $Q$ takes a snapshot, it will remain blocked and all of its inverse acquaintances will be in $Q$. Hence at $t_f$, all the actors in $Q$ are terminated.

\paragraph*{Initial configuration} Since $Q_0 = \emptyset$, the invariant trivially holds.

\paragraph*{$\textsc{Snapshot}(B, \Phi_B)$}

Suppose \(B \in Q\) takes a snapshot at time \(t\). We show that if $\Refob x A B$ is unreleased at time $t$, then $Q \vdash \Unreleased(x)$ and there are no undelivered messages along $x$ from $A$ to $B$. We do this with the help of two lemmas.

\begin{lemma}\label{lem:complete-ref}
    If $Q \vdash \Unreleased(\Refob x A B)$, then $x$ is unreleased at time $t$ and there are no undelivered messages along $x$ at time $t$. Moreover, if $t_A > t$, then there are no undelivered messages along $x$ throughout the interval $[t,t_A]$.
\end{lemma}
\begin{proof}[Proof (Lemma)]
    Since $Q$ is closed, we have $A \in Q$ and $\Phi_A \vdash \Activated(x)$. Since $B$ appears blocked, we must have $\Phi_A \vdash \SentCount(x,n)$ and $\Phi_B \vdash \RecvCount(x,n)$ for some $n$.
    
    Suppose $t_A > t$. Since $\Phi_A \vdash \Activated(x)$, $x$ is not deactivated and not released at $t_A$ or $t$. Hence, by Lemma~\ref{lem:msg-counts}, every message sent along $x$ before $t_A$ was received before $t$. Since message sends precede receipts, each of those messages was sent before $t$. Hence there are no undelivered messages along $x$ throughout $[t,t_A]$.
    
    Now suppose $t_A < t$. Since $\Phi_A \vdash \Activated(x)$, $x$ is not deactivated and not released at $t_A$. By IH 2, $A$ was blocked throughout the interval $[t_A,t]$, so it could not have sent a $\ReleaseMsg$ message. Hence $x$ is not released at $t$. By Lemma~\ref{lem:msg-counts}, all messages sent along $x$ before $t_A$ must have been delivered before $t$. Hence, there are no undelivered messages along $x$ at time $t$.
\end{proof}

\begin{lemma}\label{lem:complete-chains}
    Let $\Refob{x_1}{A_1}{B}, \dots, \Refob{x_n}{A_n}{B}$ be a chain to $\Refob x A B$ at time $t$. Then $Q \vdash \Unreleased(x)$.
\end{lemma}
\begin{proof}[Proof (Lemma)]
    Since all refobs in a chain are unreleased, we know $\forall i \le n,\ \Phi_B \not\vdash \Released(x_i)$ and so $Q \not\vdash \Released(x_i)$. It therefore suffices to prove, by induction on the length of the chain, that $\forall i \le n,\ Q \vdash \Created(x_i)$.
    
    \textbf{Base case:} By the definition of a chain, $\alpha_t(B) \vdash \Created(x_1)$, so $\Created(x_1) \in \Phi_B$.
    
    \textbf{Induction step:} Assume $Q \vdash \Unreleased(x_i)$, which implies $A_i \in Q$. Let $t_i$ be the time of $A_i$'s snapshot.
    
    By the definition of a chain, either the message $\Msg{B}{\InfoMsg(x_i,x_{i+1})}$ is in transit at time $t$, or $\alpha_t(A_i) \vdash \CreatedUsing(x_i,x_{i+1})$. But the first case is impossible by Lemma~\ref{lem:complete-ref}, so we only need to consider the latter.
    
    Suppose $t_i > t$. Lemma~\ref{lem:complete-ref} implies that $A_i$ cannot perform the $\textsc{SendInfo}(x_i,x_{i+1},A_{i+1},B)$ event during $[t,t_i]$. Hence $\alpha_{t_i}(A_i) \vdash \CreatedUsing(x_i,x_{i+1})$, so $Q \vdash \Created(x_{i+1})$.
    
    Now suppose $t_i < t$. By IH 2, $A_i$ must have been blocked throughout the interval $[t_i,t]$. Hence $A_i$ could not have created any refobs during this interval, so $x_{i+1}$ must have been created before $t_i$. This implies $\alpha_{t_i}(A_i) \vdash \CreatedUsing(x_i,x_{i+1})$ and therefore $Q \vdash \Created(x_{i+1})$.
\end{proof}

Lemma~\ref{lem:complete-chains} implies that $B$ cannot be in the root set. If it were, then by the Chain Lemma there would be a refob $\Refob y C B$ with a chain where $C$ is an external actor. Since $Q \vdash \Unreleased(y)$, there would need to be a snapshot from $C$ in $Q$ -- but external actors do not take snapshots, so this is impossible.

Since $B$ is not in the root set, there must be a chain to every unreleased refob $\Refob x A B$. By Lemma~\ref{lem:complete-chains}, $Q \vdash \Unreleased(x)$. By Lemma~\ref{lem:complete-ref}, there are no undelivered messages to $B$ along $x$ at time $t$. Since $B$ can only have undelivered messages along unreleased refobs (Lemma~\ref{lem:release-is-final}), the actor is indeed blocked.

\paragraph*{$\textsc{Send}(x,\vec y, \vec z, A,B,\vec C)$}

In order to maintain IH 2, we must show that if $B \in Q_t$ then this event cannot occur. So suppose $B \in Q_t$. By IH 1, we must have $Q \vdash \Unreleased(\Refob x A B)$, so $A \in Q$. By IH 2, we moreover have $A \not\in Q_t$ -- otherwise $A$ would be blocked and unable to send this message. Since $B$ appears blocked in $Q$, we must have $\Phi_A \vdash \SentCount(x,n)$ and $\Phi_B \vdash \RecvCount(x,n)$ for some $n$. Since $x$ is not deactivated at $t_A$ and unreleased at $t_B$, \cref{lem:msg-counts} implies that every message sent before $t_A$ is received before $t_B$. Hence $A$ cannot send this message to $B$ because $t_A > t > t_B$.

In order to maintain IH 1, suppose that one of the refobs sent to $B$ in this step is $\Refob z B C$, where $C \in Q_t$. Then in the next configuration, $\CreatedUsing(y,z)$ occurs in $A$'s knowledge set. By the same argument as above, $A \in Q \setminus Q_t$ and $\Phi_A \vdash \SentCount(y,n)$ and $\Phi_C \vdash \RecvCount(y,n)$ for some $n$. Hence $A$ cannot perform the $\textsc{SendInfo}(y,z,A,B,C)$ event before $t_A$, so $\Phi_A \vdash \CreatedUsing(y,z)$ and $Q \vdash \Created(z)$.

\paragraph*{SendInfo(y,z,A,B,C)}

By the same argument as above, $A \not\in Q_t$ cannot send an $\InfoMsg$ message to $B \in Q_t$ without violating message counts, so IH 2 is preserved.

\paragraph*{$\textsc{SendRelease}(x,A,B)$}

Suppose that $A \not\in Q_t$ and $B \in Q_t$. By IH 1, $\Refob x A B$ is unreleased at time $t$. Since $Q$ is finalized, $\Phi_A \vdash \Activated(x)$. Hence $A$ cannot deactivate $x$ and IH 2 is preserved.

\paragraph*{$\textsc{In}(A,R)$}

Since every potential inverse acquaintance of an actor in $Q_t$ is also in $Q$, none of the actors in $Q_t$ is a receptionist. Hence this rule does not affect the invariants.

\paragraph*{$\textsc{Out}(x,B,R)$}

Suppose $(\Refob y B C) \in R$ where $C \in Q_t$. Then $y$ is unreleased and $Q \vdash \Unreleased(y)$ and $B \in Q$. But this is impossible because external actors do not take snapshots.

\end{proof}

\end{document}